\documentclass[11pt]{amsart}
\usepackage{amsmath,amssymb, amsfonts, amsthm, verbatim, color}
\usepackage{mathrsfs}  
\usepackage[T1]{fontenc}
\usepackage{mathabx}
\usepackage{hyperref}
\usepackage[numbers]{natbib}
\usepackage{enumerate}

\usepackage{graphicx}

\numberwithin{equation}{section}
\numberwithin{figure}{section}
\newtheorem{theorem}{Theorem}
\newtheorem{proposition}[theorem]{Proposition} 

\newtheorem{lemma}[theorem]{Lemma}
\newtheorem{definition}[theorem]{Definition}
\newtheorem{remark}[theorem]{Remark}
\newtheorem{assumption}[theorem]{Assumption}
\numberwithin{theorem}{section}

\usepackage[dvipsnames]{xcolor}
\usepackage{tikz}
\usetikzlibrary{shapes,snakes}
\usetikzlibrary{calc}

\usetikzlibrary{shapes.misc}
\usetikzlibrary{math}

\newcommand{\bluesquare}[2]{
\draw[blue,line width=0.3mm] ($#1-(#2,#2)$) rectangle  ($#1+(#2,#2)$)
}
\newcommand{\redsquare}[2]{
\draw[red,line width=0.3mm] ($#1-(#2,#2)$) rectangle  ($#1+(#2,#2)$)
}

\newcommand*{\medcap}{\mathbin{\raisebox{1pt}{\scalebox{0.8}{\ensuremath{\bigcap}}}}}
\newcommand*{\medcup}{\mathbin{\raisebox{1pt}{\scalebox{0.8}{\ensuremath{\bigcup}}}}}

\DeclareMathOperator{\diam}{diam}

\newcommand{\ttau}{{\widetilde{\tau}}}
\newcommand{\tspec}{\widetilde{\sigma}(H_{\Lambda_l})}

\newcommand{\PI}{\text{PI}}
\newcommand{\FI}{\text{FI}}
\newcommand{\NR}{\text{NR}}

\DeclareMathOperator{\supp}{supp}

\DeclareMathOperator{\dist}{dist}
\newcommand{\ex}{\text{ex}}

\newcommand{\bR}{{\mathbb R}}

\newcommand{\bC}{{\mathbb C}}
\newcommand{\bZ}{{\mathbb Z}}
\newcommand{\bE}{{\mathbb E}}
\newcommand{\bP}{{\mathbb P}}

\newcommand{\bDx}{{\mathbb{D}_\bx}}

\newcommand{\Cgb}{\mathfrak{C}_{L,\ell}} 

\newcommand{\cV}{\mathcal{V}}

\newcommand{\boxx}{\Lambda_L(\bx)}

\newcommand{\cU}{{\mathcal{U}}}

\newcommand{\cE}{{\mathcal{E}}}

\newcommand{\cG}{\mathcal{G}}
\newcommand{\cS}{\mathcal{S}}
\newcommand{\cJ}{\mathcal{J}}
\newcommand{\cF}{\mathscr{F}}

\newcommand\Ups{\Upsilon}

\newcommand{\bx}{{\mathbf x}}
\newcommand{\by}{{\mathbf y}}
\newcommand{\bz}{{\mathbf z}}
\newcommand{\ba}{{\mathbf a}}
\newcommand{\bv}{{\mathbf v}}
\newcommand{\bu}{{\mathbf u}}
\newcommand{\bb}{{\mathbf b}}
\usepackage{xcolor}

\begin{document}
\title{An eigensystem approach to Anderson localization for multi-particle systems}
\author[B. Bringmann]{Bjoern Bringmann}
\author[D. Mendelson]{Dana Mendelson}
\thanks{DM was partially supported by NSF DMS-180069. }
\address{Department of Mathematics\\ University of California, Los Angeles\\ Los Angeles, CA, 90095}
\email{bringmann@math.ucla.edu}
\address{Department of Mathematics\\University of Chicago\\Chicago, IL, 60637 }
\email{dana@math.uchicago.edu}

\maketitle 

\begin{abstract}
This paper revisits the proof of Anderson localization for multi-particle systems. We introduce a multi-particle version of the eigensystem multi-scale analysis by Elgart and Klein, which had previously been used for single-particle systems.
\end{abstract}

\section{Introduction}

\emph{Anderson localization} \cite{Anderson58} is the absence of diffusion or transport in certain random media. Both experimental and theoretical aspects of Anderson localization do not only depend on the nature of the randomness, but also on other properties of the physical system, such as the number of particles and the strength of their interactions. In the following, the number of particles is denoted by $N$ and we distinguish three categories of physical systems:

\begin{enumerate}[(i)]
\item single-particle systems ($N=1$), 
\item multi-particle systems ($1<N<\infty$),
\item and infinite-particle systems ($N\rightarrow \infty$). 
\end{enumerate}

\subsection{Single-particle systems} 
The standard Anderson model for a single-particle system is given by 
\begin{equation}\label{intro:eq_single}
H= - \Delta + \lambda V \qquad \text{on } \ell^2(\mathbb{Z}^d). 
\end{equation}
Here, $d\geq 1$ is the spatial dimension, $\Delta$ is the centered, discrete Laplacian, $\lambda>0$ is a (large) disorder parameter, and $V= \{ V(x) \}_{x \in \mathbb{Z}^d}$ is a random potential. The assumptions on the random potential $V= \{ V(x) \}_{x \in \mathbb{Z}^d}$ commonly include the probabilistic independence of the potential at different sites and the absolute continuity of the single-site distribution with a bounded density, but may be  less or more restrictive in individual works. In the standard model \eqref{intro:eq_single}, Anderson localization may refer to either spectral or dynamical properties of $H$. The Hamiltonian $H$ is called spectrally localized if it has a pure-point spectrum and the eigenfunctions are exponentially localized. The Hamiltonian $H$ is called dynamically localized if,  for any $x\in \mathbb{Z}^d$,  the evolution $e^{-itH} \delta_x$ decays exponentially in space uniformly in time.  \\

The first proofs of Anderson localization \cite{DKS83,GMP77,KS78} were restricted to one-dimensional models. Later proofs of Anderson localization, which also treated higher-dimensional cases, are mostly based Aizenman and Molchanov's \emph{fractional moment method} \cite{AM93,AENSG06} or Fr\"{o}hlich and Spencer's \emph{multi-scale analysis} \cite{FS83,FMS85,GK01,DK89}. Both methods are based on estimates of the Green's function $G$ corresponding to the Hamiltonian $H$. The Green's function is defined by
\begin{equation}\label{intro:eq_single_Green}
G(x,y;z) := \langle \delta_x , (H-z)^{-1} \delta_y \rangle_{\ell^2(\mathbb{Z}^d)},
\end{equation}
where $x,y\in \mathbb{Z}^d$ and $z\in \mathbb{C}\backslash \sigma(H) \supseteq \mathbb{C}\backslash \mathbb{R}$  lies outside the spectrum of $H$. The Green's function $G$ is used because of its desirable algebraic properties, such as the resolvent identity. After the exponential decay of the Green's function has been proven, the spectral and dynamical localization of $H$ follow from the relationship between the Green's function and the eigenfunction correlator (cf. \cite[Chapter 7]{AW15}). The fractional moment method proceeds through direct estimates of $\mathbb{E}[ |G(x,y;z)|^s]$, where $0<s<1$. A detailed introduction can be found in the textbook \cite{AW15} or in the survey article \cite{Stolz11}. The multi-scale analysis (MSA) is more involved but often works under less restrictive conditions on the random potential. We note in particular the work of Bourgain and Kenig \cite{BK05} on localization near the edge of the spectrum for the Anderson-Bernoulli model on $\mathbb{R}^d$ for $d \geq 2$, the work of Ding and Smart \cite{DS20} which establishes the analogous result on $\mathbb{Z}^2$, and \cite{LZ19} for the extension of the previous result to $\mathbb{Z}^3$. Instead of estimating the Green's function directly, the (MSA) proceeds through an \emph{induction on scales}. To outline the induction procedure, we let $\Lambda_L := [-L,L]^d \medcap \mathbb{Z}^d$ be the discrete cube of width $2L$. The corresponding (finite-dimensional) Hamiltonian $H_{\Lambda_L}$ is defined by $H_{\Lambda_L}:= 1_{\Lambda_L} H 1_{\Lambda_L}$ and the corresponding Green's function $G_{\Lambda_L}$ is defined as in \eqref{intro:eq_single_Green} with $H$ replaced by $H_{\Lambda_L}$.  By using the second resolvent identity, estimates of the Green's function $G_{\Lambda_L}$ at scale $L$ can be reduced to estimates of the Green's function $G_{\Lambda_\ell}$ at a smaller scale $\ell$, where  $L=\ell^\gamma$ and $\gamma>1$ is fixed. Together with an initial scale estimate, which relies on Wegner estimates, this ultimately yields uniform estimates of $G_{\Lambda_L}$ in the length scale $L$. For an excellent introduction to the (MSA), we refer the reader to the lecture notes \cite{Kirsch07}. \\

More recently, Elgart and Klein \cite{EK16,EK19a,EK19} developed an \emph{eigensystem multi-scale analysis} (EMSA), which provides an alternative proof of Anderson localization.  Similar to the classical (MSA), the (EMSA) uses an induction on scales. Instead of working with the Green's function, however, the (EMSA) directly analyzes the eigensystem of the Hamiltonian on the discrete cube $\Lambda_L$. We anticipate that these methods may be more adaptable to the many-body setting, discussed further in Section \ref{sec:inf_part} below, since many statements for such systems are posed in terms of the eigenfunctions themselves.  It is perhaps of interest to note that certain aspects of the (EMSA) seem to have some similarities to techniques in nonlinear PDEs: for example, the mapping of (localized) eigenfunctions on smaller scales to (localized) eigenfunctions on larger scales (cf. \cite[Lemma 3.6]{EK16}) is similar to the stability theory for partial differential equations (see e.g. \cite[Section 3.7]{Tao06}). In contrast, the resolvent identity for the Green's function, which plays a similar role in the (MSA), seems to have no such simple analogue.

\subsection{Multi-particle systems}
We define the $N$-particle Anderson model by 
\begin{equation}\label{intro:eq_multi}
H^{(N)} = - \Delta^{(N)} + \lambda V + U. 
\end{equation}
Here, $\lambda>0$ is the disorder parameter. We define the $N$-particle Laplacian $\Delta^{(N)}$, the random potential $V\colon \bZ^{Nd} \rightarrow \bR$, and the interaction potential $U \colon \bZ^{Nd} \rightarrow \bR$ as follows: 
\begin{itemize}
\item  The $N$-particle Laplacian $\Delta^{(N)}$ is given by 
\begin{equation*}
\Delta^{(N)} \varphi( \bx) = \sum_{\substack{ \by \in \bZ^{Nd} \colon \| \by - \bx \|_{1}=1 }} \varphi(\by),
\end{equation*}
where $\varphi \colon \bZ^{Nd} \rightarrow \mathbb{C}$ and $\bx \in \bZ^{Nd}$. 
\item The (multi-particle) random potential $V \colon \bZ^{Nd} \rightarrow \bR$ is given by 
\begin{equation*}
V(\bx) = \sum_{j=1}^{N} \cV(x_j), 
\end{equation*}
where $\cV\colon \bZ^{d} \rightarrow \bR$ and $\bx =(x_1,\hdots,x_N) \in \bZ^{Nd}$. The conditions on the (single-particle) random potential $\cV = \{ \cV(u) \}_{u \in \bZ^d}$ will be detailed in Assumption \ref{assumption:V} below. 
\item Finally, 
 the interaction $U\colon \bZ^{Nd} \rightarrow \bR$ is given by 
\begin{equation*}
U(\bx) = \sum_{1\leq i < j \leq N} \cU(x_i - x_j), 
\end{equation*}
where $\mathcal{U}\colon \bZ^{d} \rightarrow \bR$ has finite support and $\bx =(x_1,\hdots,x_N) \in \bZ^{Nd}$. 
\end{itemize}
The most essential difference between the single- and multi-particle Anderson models \eqref{intro:eq_single} and \eqref{intro:eq_multi} lies in the probabilistic dependencies of the random potential. In the single-particle setting, $V(x)$ and $V(y)$ are probabilistically independent for any $x,y \in \bZ^d$ satisfying $x \neq y$. In the multi-particle setting, however, $V(\bx)$ and $V(\by)$ can be probabilistically dependent as soon as $\{ x_i \}_{i=1}^N$ and $\{ y_j \}_{j=1}^{N}$ are not disjoint, which occurs infinitely often (even for a fixed $\bx\in \bZ^{Nd}$). In fact, $V(\bx)$ and $V(\by)$ are not only probabilistically dependent but identical if $\by$ is a permutation of $\bx$. To be precise, we let $S_N$ be the symmetric group on $\{1,\hdots,N\}$. For any $\pi \in S_N$ and $\bx \in \bZ^{Nd}$, we define
\begin{equation*}
\pi \bx = ( x_{\pi(1)}, x_{\pi(2)},\hdots, x_{\pi(N)}). 
\end{equation*}
Then, it is clear from the definition of the random potential that $V(\bx)= V(\pi \bx)$ for all $\bx\in \bZ^{Nd}$ and $\pi \in S_N$. The permutation invariance is not only exhibited by the random potential $V$, but also by the $N$-particle Laplacian $\Delta^{(N)}$ and the interaction potential $U$. As a consequence of the permutation invariance,  it is natural to measure decay not in the standard  $\ell^\infty$-norm,  but instead in the symmetrized distance 
\begin{equation}\label{intro:eq_symmetrized_distance}
d_S(\bx,\by) := \min_{\bx,\by\in S_N} \| \bx - \pi \by \|_{\infty}. 
\end{equation}
The symmetrized distance $d_S$ can be viewed as a semi-metric on $\bZ^{Nd}$ or as a metric on the quotient space $\bZ^{Nd}/ S_N$. In some previous works, the symmetrized distance was replaced by the weaker Hausdorff-distance (cf. \cite{AW09}). 
Except for this change in the metric, the spectral and dynamical localization of the multi-particle Hamiltonian \eqref{intro:eq_multi} can be defined exactly as for single-particle systems. \\

Anderson localization for the multi-particle Hamiltonian \eqref{intro:eq_multi} has been proven using multi-particle versions of both the fractional moment method \cite{AW09,FW15} and multi-scale analysis \cite{CBS11,C12,C14,C15,C16,C17,C19,CS09a,CS09b,CS14}. The main idea, which was first used in \cite{AW09,CS09b}, is to prove Anderson localization for the multi-particle Hamiltonian \eqref{intro:eq_multi} through an induction on the number of particles. Informally speaking, the bounds on the Green's function or eigenfunction correlator in \cite{AW09,CS09b} distinguish the following two scenarios:  If the particle-configuration can be separated into two sub-configurations with distant particles, the desired conclusion follows from the induction hypothesis. If the particle-configuration cannot be further decomposed, then all particles need to be localized near a single site and form a ``quasi-particle'' (cf. \cite[Section 1]{AW09}). This case is then treated using single-particle methods. Of course, this informal sketch is overly simplistic and the probabilistic dependencies in the random potential turn out to be a challenge. \\ 

In this work, we present an alternative proof of Anderson localization for the $N$-particle Hamiltonian \eqref{intro:eq_multi}, which is based on the eigensystem multi-scale analysis of Elgart and Klein. Instead of striving for optimal or the most general results, we focus on accessible and simple arguments. We first describe the assumptions on the random potential which will be used in this paper.

\begin{assumption}\label{assumption:V}
Let $V= \{ V(x;\omega) \}_{x\in \bZ^d}$  be a sequence of independent, identically-distributed random variables. We assume that the single-site distribution is given by $\rho(v) \mathrm{d}v$, where the density $\rho \colon \bR \rightarrow \bR_{\geq 0}$ satisfies the following conditions:
\begin{enumerate}[(i)]
\item (Compact support) There exists a $v_{\max} \in \bR_{>0}$ such that $\operatorname{supp}(\rho) \subseteq [0,v_{\max}]$. 
\item \label{item:upper_lower_bound} (Upper/lower bound) There exist $\rho_{\min},\rho_{\max} \in \bR_{>0}$ such that 
\begin{equation*}
\rho_{\min}\leq \rho(v) \leq \rho_{\max} \qquad \text{for all } v \in [0,v_{\max}]. 
\end{equation*}
In particular, the density $\rho$ is strictly positive on its support. 
\item (Smoothness) \label{item:smoothness} The density $\rho$ is continuously differentiable on $(0,v_{\max})$ and 
\begin{equation*}
\| \rho^\prime(v) 1_{(0,v_{\max})}(v)\|_{L_v^\infty(\bR)} < \infty. 
\end{equation*}
\end{enumerate}
\end{assumption}

The assumptions on the random potential are physically reasonable. In particular, they include the uniform distribution on the compact interval $[0,v_{\max}]$. We note that Assumption \ref{assumption:V} excludes the case of Bernoulli random variables, which remains open. The full strength of Assumption \ref{assumption:V} is only used to obtain a regular conditional mean (see Definition \ref{definition:rcm} below), which may replace the lower bound in \eqref{item:upper_lower_bound} and \eqref{item:smoothness}. However, we found the current set of assumptions more accessible.  In comparison with other works on multi-particle localization, our assumptions are essentially identical to \cite{C12} but more restrictive than \cite{AW09,CS09a,CS09b}.

 In order to state our main result, we also need to introduce some additional notation. For any $\ba\in \bZ^{Nd}$ and $L>0$, we define the symmetrized $N$-particle cube by 
\begin{equation}
\Lambda_L(\ba) \equiv \Lambda_L^{(N)}(\ba) := \{ \bx \in \bZ^{Nd} \colon d_S(\bx,\ba)\leq L\big\}. 
\end{equation}
We emphasize that the side-length $L>0$ is not required to be an integer. While the cube $\Lambda_L(\ba)$ only depends on the integer part of $L$, this has notational advantages once we take fractional powers of $L$. In many statements below, where the center $\ba\in \bZ^{Nd}$ is not essential, we simply write $\Lambda_L$. Similar as for single-particle systems, we define the restriction of the Hamiltonian $H$ to the cube $\Lambda_L$ by $H_{\Lambda_L} = 1_{\Lambda_L} H 1_{\Lambda_L}$. In order to state the next definition, we let $\tau \in (0,1)$ be a parameter as in \eqref{intro:eq_parameter} below. We also denote by $\widetilde{\sigma}(H_{\Lambda_L})$ the multi-set containing the eigenvalues of $H_{\Lambda_L}$ repeated according to their multiplicities.

\begin{definition}[$m$-localizing]\label{definition:localizing}
Let $N\geq 1$, let $\Lambda_L = \Lambda_L^{(N)}(\ba)$, where $\ba \in \bZ^{Nd}$, $\bx\in \Lambda_L$, and $m>0$. Then, $\varphi \in \ell^2(\Lambda_L)$ is called $(\bx,m)$-localizing if $\| \varphi \|_{\ell^2}=1$ and 
\begin{equation*}
|\varphi(\by)| \leq \exp( - m \cdot d_S(\by,\bx) )  \qquad \forall  \, \by \in \Lambda_L\colon  d_S(\by,\bx)\geq L^\tau. 
\end{equation*}
The function $\varphi$ is called $m$-localizing if there exists an $\bx \in \Lambda_L$ such that  $\varphi$ is $(\bx,m)$-localizing. \\
Furthermore, the cube $\Lambda_L$ is called $m$-localizing for $H$ if there exists an orthonormal eigenbasis
\begin{equation*}
 (\varphi_\theta, \theta)_{\theta \in \widetilde{\sigma}(H_{\Lambda_L})} 
\end{equation*}
such that $\varphi_\theta$ is $m$-localizing for all $\theta \in \widetilde{\sigma}(H_{\Lambda_L})$. Instead of writing that $\Lambda_L$ is $m$-localizing for $H$, we sometimes simply write that $H_{\Lambda_L}$ is $m$-localizing, which places more emphasis on the Hamiltonian.
\end{definition}

Definition \ref{definition:localizing} is the multi-particle version of \cite[Definition 1.3 and Definition 1.5]{EK16}. Equipped with this definition, we can now state our main result. 

\begin{theorem}[Anderson localization]\label{theorem:AL}
Let $d=1$, let $N\geq 1$, and let $H^{(N)}$ be the $N$-particle Anderson Hamiltonian as in \eqref{intro:eq_multi}. Assume that  the localization length $m>0$, the decay parameter $p>0$, the initial length scale $L_0 \geq 1$, and the disorder parameter $\lambda>0$ satisfy 
\begin{equation*}
L_0 \geq C(m, N, p)
\qquad \text{and} \qquad 
\lambda \geq C(m, N, p, \| \rho\|_\infty, L_0), 
\end{equation*}
for some explicit constants. Then, we have for all $L\geq L_0$ that 
\begin{equation}\label{eq:AL}
\inf_{\ba \in \bZ^{Nd}} \mathbb{P} \big( \Lambda_L^{(N)}(\ba) \text{ is } m\text{-localizing for } H^{(N)} \big) \geq 1 - L^{-p}. 
\end{equation}
\end{theorem}
\begin{remark}
In particular, in the sequel, we will define parameters $\tau, \gamma, \beta$ in Section \ref{sec:not} to be used in the proof, and one may take the constants to be
\begin{equation*}
L_0 \geq C(\rho,m,\tau,\gamma,\beta,d) ~  p_\ast^{\frac{1}{\gamma\beta}+\frac{2}{1-\tau\gamma}}  
\qquad \text{and} \qquad 
\lambda \geq 2Nd \| \rho\|_\infty (1+e^{2m}) (N!)^2 (2L_0+1)^{2Nd} L_0^{p_\ast}, 
\end{equation*}
where 
\begin{equation*}
 p^\ast=  C(\gamma) \gamma^N \max(p,Nd). 
 \end{equation*}
\end{remark}

\begin{remark}
The main defect of Theorem \ref{theorem:AL} is the restriction to dimension $d=1$. This restriction is only used in the proof of the covering properties (Lemma \ref{lemma:cover}), but does not enter in other parts of the argument. The geometric difficulties occurring in dimension $d\geq 2$, which are a result of the symmetrization of the cubes, are further discussed in Section \ref{section:covers} and the appendix.  
\end{remark}

While Theorem \ref{theorem:AL} does not explicitly contain the spectral or dynamical localization of the $N$-particle Hamiltonian, both properties can be obtained as consequences of this result, as in \cite[Corollary 1.8]{EK16}. In Theorem \ref{theorem:AL}, we obtain polynomial tails in our probabilistic estimate. By working with polynomial  instead of exponential tails, we simplify the probabilistic aspects of the inductive step in the length scale (Theorem \ref{thm:ind_scale}). \\ 

The proof of Theorem \ref{theorem:AL} combines  the eigensystem multi-scale analysis of \cite{EK16,EK19a,EK19} with the induction on the number of particles from \cite{AW09,CS09b}. The main difficulties are tied to probabilistic dependencies in the random potential $V\colon \bZ^{Nd} \rightarrow \bR$, which affect the Wegner estimate (Proposition \ref{prop:wegner}) and the induction step (Theorem \ref{thm:ind_scale}).

\subsection{Infinite-particle systems} \label{sec:inf_part}

In recent years, there has been much activity and interest in Anderson localization for many-body systems ($N\rightarrow\infty$). Since a survey of this activity is well-beyond the scope of this introduction, we refer the reader to the review articles \cite{AP17,AADGHK17} in the physical literature and \cite{SS15,ANSS17} in the mathematical literature. We instead focus on the nonlinear Schr\"{o}dinger equation with a random potential given by
\begin{equation}\label{intro:RNLS}
i \partial_t q - \epsilon \Delta q + \cV q  \pm \delta |q|^2 q   = 0 \qquad (t,x) \in \mathbb{R} \times \mathbb{Z},
\end{equation}
where $\epsilon,\delta>0$. The nonlinear Schr\"{o}dinger equation describes the one-particle density of the $N$-particle Schr\"{o}dinger equation in the many-body limit (cf. \cite{Schlein13} and the references therein).  The random nonlinear Schr\"{o}dinger equation \eqref{intro:RNLS} was first studied by Bourgain and Wang \cite{BW07,BW08}, Fishman,  Krivolapov, and Soffer \cite{FKS08,FKS09}, Fr\"{o}hlich, Spencer, and Wayne \cite{FSW86}, and Wang and Zhang \cite{WZ09}. The previous results \cite{WZ09} imply that, with high probability, the evolution of any localized initial data remains localized up to the time $C_k (\epsilon+\delta)^{-k}$, where $k \geq 1$ is arbitrary. However, they do not provide detailed information on the asymptotic behavior as $t\rightarrow \infty$ for general localized initial data. During the preparation of this work, this problem was revisited by 
Cong, Shi, and Zhang \cite{CSZ20} and Cong and Shi \cite{CS20}. In \cite{CS20}, the authors show for any  localized initial data that the solutions to \eqref{intro:RNLS} satisfy
\begin{equation}
\sum_{j\in \bZ} j^2 |q_j(t)|^2 \leq C\big(\kappa, \big(q_j(0)\big)_{j\in \bZ} \big) |t|^\kappa
\end{equation}
for all $t\in \bR$, as long as $\delta$ and $\epsilon$ are sufficiently small depending on $\kappa$.  
 The proofs rely on certain normal form (or symplectic) transformations, which successively remove non-resonant terms inside the Hamiltonian. 
In particular, the methods are more closely related to the original form of the multi-scale analysis \cite{FS83} than its later implementations (see e.g. \cite{DK89,Kirsch07}). 

While we were initially motivated by the asymptotic behavior of \eqref{intro:RNLS}, we view the eigensystem multiscale analysis as a bridge between the dynamical system methods, as in \cite{CS20,CSZ20,BW07,BW08,FKS08,FKS09,FSW86,WZ09},   and more classical approaches to Anderson localization, as in
 \cite{AM93,AENSG06,FS83,FMS85,GK01,DK89}. We conclude this introduction with an interesting observation, which illustrates the differences between the two perspectives: in the multi-particle setting \cite{AW09,CS09b}, the worst dependence on the particle number $N$ in proofs of localization is a result of summing over all possible particle configurations inside the cube. In contrast, the worst-case in proofs of localization for \eqref{intro:RNLS} is a ``quasi-soliton'', which only has a few degrees of freedom (cf. \cite[p.~845]{FKS08}). 

\medskip 
\subsection{Notation} \label{sec:not}
We first choose the parameters that will be needed throughout this paper. We fix $\beta,\tau \in (0,1)$ and $\gamma>1$ such that 
\begin{equation}\label{intro:eq_parameter}
\beta < \gamma^{-1} <1 < \gamma < 2 \qquad \text{and} \qquad \max\Big( \gamma \beta, \gamma^{-1} \Big) < \tau <1. 
\end{equation}
Turning towards the geometric objects in this paper, we 
let $\Phi \subset \Theta \subset \bZ^{Nd}$ be symmetric sets and define the boundary 
\[
\partial^{\Theta} \Phi = \{ (u,v) \in \Phi \times (\Theta \setminus \Phi) \,:\, d_S(u,v) = 1\},
\]
the exterior boundary
\[
\partial_{ex}^{\Theta} \Phi = \{ v \in \Theta \setminus \Phi \,:\, \exists u \,\, s.t.\,\, (u,v) \in \partial^{\Theta} \Phi \},
\]
and the interior boundary
\[
\partial_{in}^{\Theta} \Phi = \{ u \in \Theta \setminus \Phi \,:\, \exists v \,\, s.t.\,\, (u,v) \in \partial^{\Theta} \Phi \}.
\]
For any $r \geq 1$, we also define
\begin{equation*}
\Phi^{\Theta,r} = \big \{ \bx \in \Phi \colon d_S(\bx, \Theta\backslash \Phi) \geq r \big\}.
\end{equation*}
Thus, $\Phi^{\Theta,r}$ contains all particle-configurations in $\Phi$ which have a symmetrized distance of at least $r$ to all particle configurations in $\Theta \backslash \Phi$. The sets $\partial_{ex}^{\Theta} \Phi$, $\partial_{\text{in}}^{\Theta} \Phi$, and $\Phi^{\Theta,r}$ are illustrated in Figure \ref{figure:geometric} below. \\
Furthermore, let $\bx \in \bZ^{Nd}$ and let $1\leq n \leq N$. We define the projection onto the $n$-th particle by
\begin{equation}
\Pi_n \bx = x_n. 
\end{equation}
We also denote the sites corresponding to the particle configurations in $\Theta \subset \bZ^{Nd}$ by 
\begin{equation}\label{eq:Pi}
\Pi \Theta = \bigcup_{n=1}^N \Pi_n\Theta.
\end{equation}
For any site $u \in \bZ^{d}$, we define the number operator $N_u \colon \bZ^{Nd} \rightarrow \mathbb{N}$ by
\begin{equation}\label{eq:number_operator}
N_u(\bx) = \# \{ 1 \leq n \leq N\colon x_n = u \}. 
\end{equation}
Finally, we recall the geometric decomposition of the Hamiltonian given by
\begin{equation}\label{eq:geometric_decomposition}
H_\Theta = ( H_\Phi \oplus H_{\Theta \backslash \Phi} ) + \Gamma_{\partial^\Theta \Phi}. 
\end{equation}
Here, $\oplus$ denotes the direct sum and the boundary term $\Gamma_{\partial^\Theta \Phi}$ is defined by 
\begin{equation*}
\big( \Gamma_{\partial^\Theta \Phi} \delta_\bx  \big) (\by) = 
\begin{cases} 
\begin{tabular}{ll}
$1$ &if $(\bx,\by)\in \partial^{\Theta} \Phi $ or $(\by,\bx)\in \partial^{\Theta} \Phi$, \\
$0$ & else.
\end{tabular}
\end{cases} 
\end{equation*}

\begin{figure}\label{figure:boxes}
\begin{tikzpicture}[scale=0.8]
\draw[<->, ultra thick] (0,6)--(0,0)--(6,0);

\node at (6,-0.5) {\small$x_1$};
\node at (-0.6,5.9) {\small$x_2$};

\draw [draw=white, fill=red, opacity=0.1] (1,1) rectangle (3,3);
\node at (2.65,2.65) {\Large \textcolor{red}{$\Phi$}};
\node at (4,2) {\Large \textcolor{red}{$\partial_{\text{in}}^{\Theta}\Phi$}};

\draw [draw=red, ultra thick] (3,1) -- (3,3) -- (1,3);

\draw [draw=OliveGreen, ultra thick] (3.1,1) -- (3.1,3.1) -- (1,3.1);
\node at (2,3.75)  {\Large \textcolor{OliveGreen}{$\partial_{\text{ex}}^{\Theta}\Phi$}};

\draw [draw=blue, ultra thick] (1,1) rectangle (5,5);
\node at (3,5.5) {\Large \textcolor{blue}{$\Theta$}};

\begin{scope}[xshift=9cm]
\draw[<->, ultra thick] (0,6)--(0,0)--(6,0);

\node at (6,-0.5) {\small$x_1$};
\node at (-0.6,5.9) {\small$x_2$};


\node at (2.65,2.65) {\Large \textcolor{red}{$\Phi$}};
\draw [draw=red, ultra thick] (3,1) -- (3,3) -- (1,3);

\draw [draw=violet, ultra thick, fill=violet, fill opacity=0.1] (1,1) rectangle (2.35,2.35);
\node at (1.7,1.5)  {\Large \textcolor{violet}{$\Phi^{\Theta,r}$}};

\draw [draw=violet, thick, <->] (2.35,1.75)--(3,1.75);
\node at (2.675,1.5) {\large \textcolor{violet}{$r$}};

\draw [draw=blue, ultra thick] (1,1) rectangle (5,5);
\node at (3,5.5) {\Large \textcolor{blue}{$\Theta$}};

\end{scope}

\end{tikzpicture}
\caption{\small{We choose $d=1$, $N=2$,  $\Theta=\Lambda_L^{(2)}(\ba)$, and $\Phi=\Lambda_{L/2}^{(2)}(\bb)$, where $L>0$ is unspecified. The centers  $\ba,\bb\in \bZ^2$ are chosen so that the lower-left corners of $\Theta$ and $\Phi$ coincide. The interior and exterior boundary are illustrated in the left diagram. The set $\Phi^{\Theta,r}$ with $r=L/4$ is illustrated in the right diagram. The boundary $\partial^{\Theta} \Phi$ is not illustrated here, since it is a subset of $\bZ^2 \times \bZ^2$ rather than $\bZ^2$.}}
\label{figure:geometric}
\end{figure}
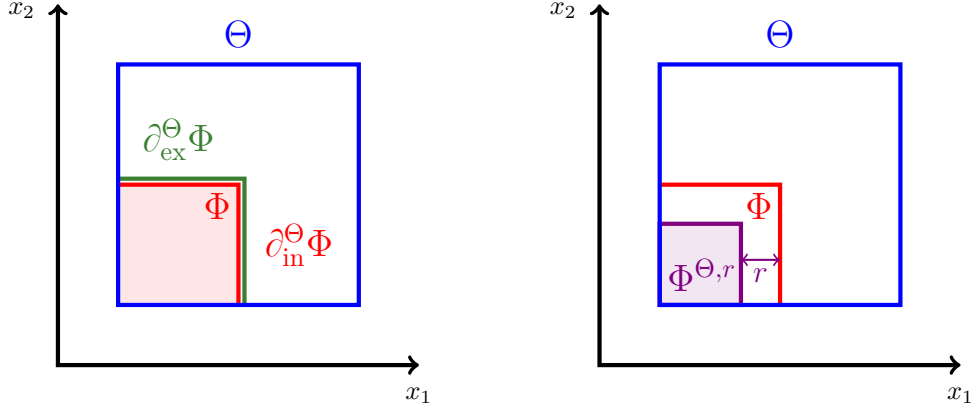

\textbf{Acknowledgements:} The authors thank Alexander Elgart, Fei Feng, Rupert Frank, Abel Klein, and Terence Tao for interesting and helpful discussions. D. M. also thanks Carlos Kenig, Charles Smart, and Wilhelm Schlag for prior conversations about Anderson localization.

\section{Multi-particle Wegner estimates}

In this section, we prove a minor variant of the multi-particle Wegner estimate from \cite[Theorem 3]{C14}. In order to state the result, we first recall the notion of spectral separation for a family of sets in $\bZ^{Nd}$. The proof of the multi-particle Wegner estimate, Proposition \ref{prop:wegner}, requires additional notation and is therefore placed in a separate subsection. 

\begin{definition}[Spectral separation]\label{definition:spectral_separation}
Let $L \geq 1$. We call two symmetric, finite sets \( \Theta_1,  \Theta_2 \in \bZ^{Nd} \) spectrally \( L\)-separated for \(H \) if 
\[
\dist(\sigma(H_{\Theta_1}),\sigma(H_{\Theta_2})) \geq 1/2 \cdot e^{-L^\beta}.
\]
 A family \( \{ \Theta_j \}_{j\in J} \) of finite, symmetric sets is called spectrally \( L\)-separated for \( H \) if \( \Theta_j \) and \( \Theta_{j^\prime} \) are spectrally \( L \)-separated for all \(j,j^\prime \in J \) such that 
 \begin{equation*}
 d_S(\Theta_j,\Theta_{j^\prime}) \geq 8N \max \big( \diam_S \Theta_j , \diam_S \Theta_{j^\prime}\big). 
 \end{equation*}
\end{definition}

We note that the spectral separation $\dist( \sigma(H_{\Theta_j}), \sigma(H_{\Theta_{j^\prime}})) \geq 1/2\, \cdot  \exp(-L^\beta)$ clearly fails if the two sets $\Theta_j$ and $\Theta_{j^\prime}$ coincide, which makes a condition on the distance between $\Theta_j$ and $\Theta_{j^\prime}$ natural. We will eventually consider families of cubes with side-lengths comparable to $\ell$, see e.g. the proof of Theorem \ref{thm:ind_scale}. 

\begin{proposition}[Multi-particle Wegner estimate]\label{prop:wegner}
Let $L,\ell \geq 1$ satisfy $L=\ell^\gamma$ and let the random potential satisfy Assumption \ref{assumption:V}. Let $\cJ$ be a finite index set and let $\{ \Theta_j \}_{j \in \cJ}$ be a family of finite, symmetric sets satisfying $\diam \Theta_j \leq 20 N \ell$ for all $j\in \cJ$. Then
\begin{equation}\label{eq:spectral_separation}
\bP \big( \{ \Theta_j \}_{j \in \cJ} \text{ is not spectrally } L \text{-separated}\big) \leq C(\rho)^{Nd} (\# \cJ)^2  N^{12Nd} \ell^{6Nd} e^{-\frac{1}{2} L^\beta}. 
\end{equation}
\end{proposition}

Before we start with the proof of Proposition \ref{prop:wegner}, we include a short discussion of \eqref{eq:spectral_separation}. If the family $\{ \Theta_j \}_{j \in \cJ}$ consists only of two symmetric sets $\Theta_1$ and $\Theta_2$ satisfying $d_S(\Theta_1,\Theta_2) \geq 80 N^2 \ell$, then \eqref{eq:spectral_separation} estimates the probability
\begin{equation}\label{eq:disc_wegner_1}
\bP\Big( \dist(\sigma(H_{\Theta_1}), \sigma(H_{\Theta_2})) \geq 1/2 \cdot \exp(-L^\beta) \Big). 
\end{equation}
However, the standard Wegner estimate (as in \cite{Kirsch07,Kirsch08,Stollmann00,Wegner81}) only controls
\begin{equation}\label{eq:disc_wegner_2}
\sup_{E \in \bR} \bP\Big( \dist(E, \sigma(H_{\Theta})) \geq 1/2 \cdot \exp(-L^\beta) \Big). 
\end{equation}
In the single-particle case, estimates of \eqref{eq:disc_wegner_2} directly lead to estimates of \eqref{eq:disc_wegner_1}, since the condition $d(\Theta_1,\Theta_2)\geq 80\ell$ implies the probabilistic independence of $H_{\Theta_1}$ and $H_{\Theta_2}$. In the multi-particle case, however, the Hamiltonians $H_{\Theta_1}$ and $H_{\Theta_2}$ can be probabilistically dependent even for an arbitrarily large symmetric distance of $\Theta_1$ and $\Theta_2$. As a result, estimates of \eqref{eq:disc_wegner_1} are more difficult than estimates of \eqref{eq:disc_wegner_2}.

\subsection{Regular conditional means, weak separability, and applications}

In this subsection, we prove Proposition \ref{prop:wegner}. Since the methods here will not be used in the rest of the paper, we encourage the reader to skip this subsection during the first reading.
Given a finite set $\cS\subseteq \bZ^d$, we denote the sample mean of the random potential on $\cS$ by 
\begin{equation}
\langle \cV \rangle_\cS := \frac{1}{\# \cS} \sum_{u\in \cS} \cV(u). 
\end{equation}
For any site \( u \in \cS \), we further denote the fluctuations of \( \cV \) relative to the sample mean by 
\begin{equation}
\eta_{u,\cS} := \cV(u)- \langle  \cV \rangle_\cS.
\end{equation}
We denote the \( \sigma\)-algebra generated by the fluctuations \( \{ \eta_{u,\cS} \}_{u\in \cS} \) by \( \mathscr{F}_\cS\). Finally, we denote by \(  F_S( \cdot| \mathscr{F}_\cS) \) the conditional distribution function of \( \langle \cV \rangle_\cS \) given \( \mathscr{F}_\cS\), i.e., 
\begin{equation*}
F(s|\cF_\cS) := \bP( \langle \cV \rangle_\cS \leq s | \cF_\cS ). 
\end{equation*}
\begin{definition}[RCM]\label{definition:rcm}
We say that the random potential \( \cV \) has a regular conditional mean if there exist constants $A_1,A_2 \in [0,\infty)$ and  \( b_1,b_2,C_1,C_2\in (0,\infty) \) such that for any finite subset \( \cS \subseteq \bZ^d \), the conditional distribution function \( F(\cdot|\cF_\cS) \) satisfies 
\begin{equation}\label{eq:rcm}
\bP \Big( \sup_{t\in \bR} | F(t+2s|\cF_\cS) - F(t|\cF_\cS)| \geq C_1 (\# \cS)^{A_1} s^{b_1}\Big) \leq C_2 (\# \cS)^{A_2} s^{b_2} 
\end{equation}
for all $s>0$. 
\end{definition}

The simplest example of a random potential $\cV$ with a regular conditional mean is a standard Gaussian potential. In that case, the sample mean $\langle \cV \rangle_{\cS}$ is a Gaussian random variable with variance $(\# \cS)^{-1}$ which is independent of $\cF_\cS$. Thus, 
\begin{equation*}
 \sup_{t\in \bR} | F(t+2s|\cF_\cS) - F(t|\cF_\cS) | \leq \sqrt{\frac{2}{\pi}} (\# \cS)^{\frac{1}{2}} s. 
\end{equation*}
In particular, \eqref{eq:rcm} holds with $C_1= \sqrt{2/\pi}, A_1 = 1/2, b_1=1$, and any choice of $C_2,A_2,$ and $b_2$. In the next lemma, we show that Assumption \ref{assumption:V} also leads to a regular conditional mean.

\begin{lemma}\label{lemma:assumption_implies_rcm}
Assume that the random potential satisfies Assumption \ref{assumption:V}. Then, there exists a constant $C=C(\rho)>0$, depending only on the single-site density $\rho$, such that 
\begin{equation*}
\bP\big( \sup_{t\in \bR} | F(t+2s|\cF_\cS) - F(t|\cF_\cS) | > s^{\frac{1}{2}} \big)\leq C(\rho) (\# \cS)^6 s 
\end{equation*}
for all $\cS\subseteq \bZ^d$ and $s>0$. In particular, \eqref{eq:rcm} is satisfied with 
\begin{equation*}
C_1=1, ~ A_1=0,~  b_1=1/2, ~  C_2=C(\rho), ~ A_2=6, \quad \text{and} \quad b_2=1.
\end{equation*}
\end{lemma} 

\begin{proof}
This is a consequence of \cite[Theorem 4 with $\alpha=1/2$]{Chulaevsky13}. The restriction on $s$ in \cite[Theorem 4]{Chulaevsky13} is circumvented by increasing the power of $N$ in \cite[(6.5)]{Chulaevsky13}, which makes the bound trivial in the excluded range. 
\end{proof}

This finishes our discussion of Definition \ref{definition:rcm} and its relation to Assumption \ref{assumption:V}. We now turn to the more geometric aspects of the Wegner estimates. 

\begin{definition}\label{definition:weakly_separable}
We call two symmetric sets \( \Theta_1,\Theta_2 \subseteq \bZ^{Nd} \) weakly separable if there exists a set of sites \( \cS \subseteq \bZ^d \) and \(  0 \leq N_1 \neq N_2 \leq N \) such that
\begin{equation}\label{eq:weakly_separable}
\# \{ 1 \leq j \leq N \colon x_j \in \cS \}= N_1 \quad \text{and} \quad \# \{ 1\leq j \leq N \colon y_j \in \cS\}=N_2 \quad \text{for all } (\bx,\by)\in \Theta_1\times \Theta_2. 
\end{equation}
\end{definition}

In other words, $\Theta_1$ and $\Theta_2$ are called weakly separable if we can find a region $\cS\subseteq \bZ^d$ in which the number of particles is constant for each set $\Theta_1$ and $\Theta_2$ and differs between the two sets. An example of two weakly separable sets is given by the two cubes $\Lambda_L(\bx)$ and $\Lambda_L(\by)$ in Figure \ref{figure:two_particle_cubes}. In that example, \eqref{eq:weakly_separable} holds with $\cS=\{3,4,5\}$, $N_1=1$, and $N_2=2$.  

\begin{figure}
\begin{tikzpicture}[scale=0.8]
\draw[->, ultra thick] (0,0)--(6,0);
\draw[->, ultra thick] (0,0)--(0,6);

\node at (6,-0.5) {\small$x_1$};
\node at (-0.6,5.9) {\small$x_2$};
\node at  (4, -0.4) {\small $4$} ;
\node at  (-0.4, 4) {\small $4$} ;
\node at  (-0.4, 1) { \small$1$} ;
\node at  (1, -0.4) { \small$1$} ;
\draw[thick] (1,0.15)--(1,-0.15);
\draw[thick] (0.15,1)--(-0.15,1);
\draw[thick] (4,0.15)--(4,-0.15);
\draw[thick] (0.15,4)--(-0.15,4);

\draw [fill=blue, draw=black, fill opacity=0.3] (3,3)--(3,5)--(5,5)--(5,3)--(3,3);
\draw [fill=orange, draw=black, fill opacity=0.3] (3,0)--(3,2)--(5,2)--(5,0)--(3,0);
\draw [fill=orange, draw=black, fill opacity=0.3] (0,3)--(2,3)--(2,5)--(0,5)--(0,3);

\draw[fill=black] (4,1) circle (0.5mm);
\node[below] at (4.2,1) {\small$\bx$};
\draw[fill=black] (1,4) circle (0.5mm);
\node[below] at (1.2,4) {\small$\pi\bx$};
\draw[fill=black] (4,4) circle (0.5mm);
\node[below] at (4.2,4) {\small$\by$};
\end{tikzpicture}
\caption{\small{We display two-particle cubes in one spatial dimension. We let \( \bx = (4,1) \), \( \by = (4,4) \), and \( L=1 \). If \( \pi \in S_2 \) is the permutation given by \( \pi(1)=2 \) and \( \pi(2)=1 \), we have \( \pi \bx = (1,4) \). The two-particle cube \( \Lambda_L^{(2)}(\bx) \) corresponds to the orange area  and the two-particle cube \( \Lambda_L^{(2)}(\by) \) corresponds to the blue area.} }
\label{figure:two_particle_cubes}
\end{figure}
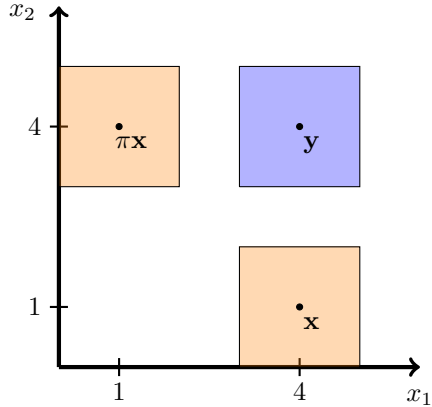

\medskip
In order to use Definition \ref{definition:weakly_separable} in our proof of multi-particle localization, we clarify the relation between weak separability and the symmetric distance of sets, which is done in the next lemma.

\begin{lemma}\label{lemma:distance_weakly_separable}
If \( \Theta_1,\Theta_2\subseteq \bZ^{Nd} \) are symmetric and satisfy 
\begin{equation}
d_S(\Theta_1,\Theta_2) \geq 8N \max(\diam_S (\Theta_1),\diam_S(\Theta_2)),
\end{equation}
then \( \Theta_1\) and \( \Theta_2 \) are weakly separable. 
\end{lemma}

This is a slight generalization of \cite[Lemma 2.1]{C14}. 

\begin{proof}
The main idea behind the proof is that \( \bx,\by \in \bZ^{Nd} \) satisfy \( d_S(\bx,\by)=0 \) if and only if \( N_u(\bx)=N_u(\by) \) for all \( u \in \bZ^d \). We now present a more stable version of this elementary fact. \\
To simplify the notation, we set \( L:= \max( \diam_S(\Theta_1),\diam_S(\Theta_2)) \). We also fix two particle-configurations  \( \ba \in \Theta_1 \) and \( \bb \in \Theta_2\). We denote by \( \Gamma_1,\hdots,\Gamma_M \), where \( 1 \leq M \leq N \), the connected components of 
\begin{equation}
\Big(\bigcup_{j=1}^{J} \Lambda_{L}(a_j)\Big) \medcup  \Big(\bigcup_{j=1}^{J} \Lambda_{L}(b_j)\Big) .
\end{equation}
For each \( 1 \leq m\leq M \), we let \( \cJ_m(\ba) \) and \( \cJ_m(\bb) \) be the indices corresponding to the connected component \( \Gamma_m \). Since $\Gamma_m$ is connected, we obtain that 
\begin{equation*}
\diam \Gamma_m \leq \sum_{j\in \cJ_m(\ba)} \diam \Lambda_L(\ba_j) +  \sum_{j\in \cJ_m(\bb)} \diam \Lambda_L(\bb_j) \leq 4 N L. 
\end{equation*}
We now claim that there exists an \( 1\leq m\leq M \) such that 
\begin{equation}\label{eq:Jm}
\# \cJ_m(\ba) \neq \# \cJ_m(\bb). 
\end{equation}
If not, then \( \ba \) and \( \bb \) are close in the symmetrized distance. More precisely, it holds that
\begin{equation*}
d_S(\Theta_1,\Theta_2)\leq d_S(\ba,\bb) \leq \max_{m=1,\hdots,M} \| \ba_{\cJ_m(\ba)} - \bb_{\cJ_m(\bb)} \| \leq \max_{m=1,\hdots,M} \diam( \Gamma_m) \leq 4NL,
\end{equation*}
which is a contradiction. For any \( m \) as in \eqref{eq:Jm}, we define \( \cS:= \Gamma_m \), \( N_1 := \# \cJ_m(\ba) \), and \( N_2 := \# \cJ_m(\bb) \). It remains to verify the conditions in Definition \ref{definition:weakly_separable}. Since the argument for \( \Theta_2 \) is similar, it suffices to check for all \( \bx \in \Theta_1 \) that 
\begin{equation}\label{eq:occupation_x}
\# \{ 1 \leq j \leq N \colon x_j \in \cS\} = N_1 .
\end{equation}
Due to the permutation invariance of \eqref{eq:occupation_x}, we may assume that \( \| \bx - \ba \|_{\ell^\infty} \leq \diam_S(\Theta_1) \leq L \). It then follows for all \( 1 \leq j \leq N \) that \( x_j \in \Lambda_L(a_j) \). Due to the equivalence
\begin{equation*}
\Lambda_L(a_j) \medcap \Gamma_m \neq \emptyset \quad \Leftrightarrow \quad \Lambda_L(a_j) \subseteq  \Gamma_m \quad \Leftrightarrow \quad j \in \cJ_m(\ba),  
\end{equation*}
it follows that \( x_j \in \Gamma_m \) if and only if \( j \in \cJ_m(\ba) \). This yields
\begin{equation*}
\# \{ 1 \leq j \leq N \colon x_j \in \cS\}  = \# \{ 1 \leq j \leq N \colon x_j \in \Gamma_m \}= \# \cJ_m(\ba) = N_1,
\end{equation*} 
which yields \eqref{eq:occupation_x} and completes the proof. 
\end{proof}

The following proposition is a minor extension of \cite[Theorem 3]{C14}. 
\begin{proposition}[Basic multi-particle Wegner estimate]\label{prop:rcm_wegner}
Let \( V \) be a random potential with a regular conditional mean as in Definition \ref{definition:rcm} with constants $A_i, b_i, C_i$ for $i=1,2$. Let \( \Theta_1, \Theta_2 \subseteq \bZ^{Nd} \) be two finite and symmetric sets with 
\begin{equation}
 d_S(\Theta_1,\Theta_2) \geq 8N \max(\diam_S(\Theta_1),\diam_S(\Theta_2)).
\end{equation}
For any \( s >0 \), it then holds that 
\begin{equation}
\begin{aligned}
&\bP( \dist( \sigma( H_{\Theta_1}),\sigma( H_{\Theta_2}))\leq s )\\
& \leq C_1 (\# \Theta_1) (\# \Theta_2) \# (\Theta_1 \medcup \Theta_2)^{A_1} (2s)^{b_1} + C_2 \# (\Theta_1 \medcup \Theta_2)^{A_2} (2s)^{b_2},
\end{aligned}
\end{equation}
\end{proposition}

\begin{proof}
By replacing \cite[Lemma 8]{C14} with Lemma \ref{lemma:distance_weakly_separable}, the argument extends essentially verbatim to our situation. The role of $L$ in \cite[Theorem 3]{C14} is taken by $\# \cS$ in Definition \ref{definition:weakly_separable}. Since $\cS$ can always be taken as a subset of the projections $\Pi \Theta_1 \medcup \Pi \Theta_2$, we have $\# \cS \leq \# (\Theta_1 \medcup \Theta_2)$. 
\end{proof}

We are now ready to prove the main result of this section.

\begin{proof}[Proof of Proposition \ref{prop:wegner}:]
Using that $\diam \Theta_j \leq 20 N \ell$ for all $j \in \cJ$, we obtain for an absolute constant $C\geq 1$ that 
\begin{equation*}
\max_{j \in \cJ} \# \Theta_j  \leq N! (40N\ell+1)^{Nd} \leq  C^N N^{2Nd} \ell^{Nd}. 
\end{equation*}

 Now, let $j_1,j_2 \in \cJ$ be such that 
\begin{equation}
 d_S(\Theta_{j_1},\Theta_{j_2}) \geq 8N \max \big( \diam_S \Theta_{j_1} , \diam_S \Theta_{j_2}\big). 
\end{equation}
By combining Lemma and  Proposition \ref{prop:rcm_wegner}, we obtain that 
\begin{align*}
\bP\Big( \dist( \sigma( H_{\Theta_{j_1}}),\sigma( H_{\Theta_{j_2}}))\leq \frac{1}{2} e^{-L^\beta} \Big) 
& \leq  (\# \Theta_{j_1}) (\# \Theta_{j_2}) e^{- \frac{1}{2} L^\beta}  +  C(\rho) \# (\Theta_{j_1} \medcup \Theta_{j_2})^{6 } e^{- L^\beta}  \\
& \leq C(\rho)^{Nd} N^{12Nd} \ell^{6Nd} e^{-\frac{1}{2} L^\beta} . 
\end{align*}
The desired estimate now follows from a union bound. 
\end{proof}
\section{Preparations}

We provide the tools needed in the  eigensystem multi-scale analysis in Section \ref{section:inductive}.

\subsection{Covers}\label{section:covers}

As is evident from its name, the eigensystem multiscale analysis connects multiple scales. The notion of a cover allows us to decompose boxes at a large scale $L$ into several boxes at a smaller scale $\ell$. 

\begin{definition}\label{definition:cover}
Let \( \Lambda_L^{(N)}(\bb) \), \( \bb \in \bZ^{Nd} \), be an \( N \)-particle box. We define the cover \( C_{L,\ell} = C_{L,\ell}^{(N)}(\bb) \) by 
\begin{equation}
C_{L,\ell} = \big\{ \Lambda_\ell^{(N)}(\ba) \colon \ba \in \bZ^{Nd}  \text{ and } \Lambda_\ell^{(N)}(\ba)\subseteq \Lambda_L^{(N)}(\bb) \big\}. 
\end{equation}
We further denote the centers of the cubes in \( C_{L,\ell} \) by \( \Xi_{L,\ell} \). 
\end{definition}
In the one-particle version of the eigensystem multiscale analysis \cite{EK19}, the authors rely on \emph{suitable} covers (cf. \cite{EK19,GK06}), which contain fewer boxes but still satisfy similar covering properties. Since we are not optimizing the different parameters in our argument, the simpler notion of a cover as in Definition \ref{definition:cover} is sufficient for our purpose. For more elementary arguments based on covers, we refer to the excellent lecture notes by Kirsch \cite[Section 9]{Kirsch07}. 

\begin{remark}\label{rem:1d}
Unfortunately, the properties of covers in the (symmetrized) multi-particle setting are much more complicated than in the one-particle setting. While we obtained all the necessary properties for the multi-particle multi-scale analysis in one spatial dimension, i.e., \( d=1 \), we were unable to solve the geometric difficulties in dimension \( d \geq 2 \). Indeed, our argument uses the non-decreasing rearrangement of particles in $\bZ$, which has no direct analogue in higher spatial dimensions.
\end{remark}

We record the properties of covers used in the rest of the paper in the next lemma. In order to not interrupt the flow of the main argument, we moved the proof into a separate appendix.

\begin{lemma}[Properties of covers]\label{lemma:cover}
Let \( d=1\), let \( C_{L,\ell} \), and let \( \Xi_{L,\ell} \) be as in Definition \ref{definition:cover}. Then, we have the following properties: 
\begin{enumerate}[(i):]
\item The cardinality of the cover is bounded by \( \# C_{L,\ell} \leq (2L+1)^{Nd} \). \label{item:cover_1}
\item  \label{item:cover_2} We have that
\begin{equation}\label{eq:cover}
\Lambda_L^{(N)}(\bb) = \bigcup_{a\in \Xi_{L,\ell}} \Lambda_\ell^{\Lambda_L(\bb),\ell}(\ba). 
\end{equation}
\end{enumerate}
\end{lemma}

\subsection{Partially and fully interactive cubes}
As explained in the introduction, one of the main differences between the one-particle and multi-particle setting lies in the probabilistic independence of certain Hamiltonians. If $N=1$ and $\Theta_1,\Theta_2 \subseteq \bZ^d$ satisfy $d(\Theta_1,\Theta_2)>0$, which is equivalent to $\Theta_1 \medcap \Theta_2 \not = \emptyset$ , then the restricted Hamiltonians $H_{\Theta_1}$ and $H_{\Theta_2}$ are probabilistically independent. If $N\geq 2$ and $\Theta_1,\Theta_2 \subseteq \bZ^{Nd}$, however, the symmetrized distance $d_S(\Theta_1,\Theta_2)$ can be arbitrarily large while the one-particle sites $\Pi \Theta_1$ and $\Pi \Theta_2$, as defined in \eqref{eq:Pi}, are overlapping. As a result, the restricted Hamiltonians $H_{\Theta_1}$ and $H_{\Theta_2}$  are not probabilistically independent. In this section, we introduce partially and fully interactive cubes, which allow us to isolate the related difficulties. This notion appeared in earlier work of Chulaevsky and Suhov \cite{CS09b}.

\begin{definition}[Partially and fully interactive cubes]\label{definition:partially_interacting}
We call a cube \( \Lambda^{(N)}_L(\bx) \) partially interactive if there exist \( 1 \leq N_1,N_2 < N \), where \( N_1 + N_2 = N \), and disjoint sets \( \cS_1, \cS_2 \subseteq \bZ^{d} \) satisfying the following conditions:
\begin{itemize}
\item \( \dist( \cS_1,\cS_2) \geq C_{\cU} \), where $C_{\cU} = \max\limits_{u \in \supp \cU} \| u \|+1 $, 
\item for all \( \by \in \Lambda_L^{(N)}(\bx) \), it holds that 
\begin{equation*}
\# \{ 1 \leq j \leq N\colon y_j \in \cS_1\} = N_1 \qquad \text{and} \qquad \# \{ 1 \leq j \leq N \colon y_j \in \cS_2 \} = N_2. 
\end{equation*}
\end{itemize}
Conversely, we call  a cube  \( \Lambda^{(N)}_L(\bx) \) fully interactive if it is not partially interactive.
\end{definition}

An example of a partially interactive and a fully interactive cube is displayed in Figure \ref{figure:pi}.

\begin{figure}\label{figure:pi}
\begin{tikzpicture}[scale=0.6]
\draw[->, ultra thick] (0,0)--(7,0);
\draw[->, ultra thick] (0,0)--(0,7);

\node at (7,-0.5) {\small$x_1$};
\node at (-0.6,6.9) {\small$x_2$};
\node at  (5, -0.4) {\small $5$} ;
\node at  (-0.4, 5) {\small $5$} ;
\node at  (-0.4, 1) { \small$1$} ;
\node at  (1, -0.4) { \small$1$} ;
\draw[thick] (1,0.15)--(1,-0.15);
\draw[thick] (0.15,1)--(-0.15,1);
\draw[thick] (5,0.15)--(5,-0.15);
\draw[thick] (0.15,5)--(-0.15,5);

\begin{scope}[xshift=1cm]
\draw [fill=orange, draw=black, fill opacity=0.3] (3,0)--(3,2)--(5,2)--(5,0)--(3,0);
\end{scope}
\begin{scope}[yshift=1cm]
\draw [fill=orange, draw=black, fill opacity=0.3] (0,3)--(2,3)--(2,5)--(0,5)--(0,3);
\end{scope}
\draw[fill=black] (5,1) circle (0.5mm);
\node[below] at (5.2,1) {\small$\bx$};
\draw[fill=black] (1,5) circle (0.5mm);
\node[below] at (1.2,5) {\small$\pi\bx$};
\draw[fill=black] (5,5) circle (0.5mm);
\node[below] at (5.2,5) {\small$\by$};

\draw[<->,red, ultra thick] (-1,0)--(-1,2);
\node at (-1.5,1) {\textcolor{red}{\small$\cS_1$}};
\draw[<->,red, ultra thick] (-1,4)--(-1,6);
\node at (-1.5,5) {\textcolor{red}{\small$\cS_2$}};
\draw[<->,red, ultra thick] (0,-1)--(2,-1);
\node at (1,-1.5) {\textcolor{red}{\small$\cS_1$}};
\draw[<->,red, ultra thick] (4,-1)--(6,-1);
\node at (5,-1.5) {\textcolor{red}{\small$\cS_2$}};

\draw [fill=blue, draw=black, fill opacity=0.3] (4,4)--(4,6)--(6,6)--(6,4)--(4,4);
\end{tikzpicture}
\caption{\small{We display a partially interactive and a fully interactive two-particle cube in orange and blue, respectively. We let $d=1$, $N=2$, \( \bx = (5,1) \), and \( L=1 \). We also assume nearest neighbor interactions, that is, $\operatorname{supp}(\cU)\subseteq \{-1,0,1\}$. Then, the cube $\Lambda_L(\bx)$ is partially interactive. For instance, one can take $\cS_1=\{0,1,2\}$ and $\cS_2=\{4,5,6\}$, which satisfy $d(\cS_1,\cS_2)\geq 2$. It is clear from the picture that each particle lies in exactly one of the sets $\cS_1$ and $\cS_2$. Conversely, if $\by=(5,5)$, then the cube $\Lambda_L(\by)$ is fully interactive, since both particles separately range over the set $\{ 4,5,6\}$. }}
\end{figure}

\begin{lemma}\label{lemma:fully_interactive}
Let \( L >  C_{\cU} \) and let \( \Lambda_L^{(N)}(\bx) \), \( \Lambda_L^{(N)}(\by) \) be two fully interactive \( N\)-particle cubes. If  \( d_S(\bx,\by) \geq  8NL \), then \( \Pi \Lambda_L^{(N)}(\bx) \) and \( \Pi \Lambda_L^{(N)}(\by) \) are disjoint subsets of \( \bZ^d \). 
\end{lemma}

\begin{proof}
We first obtain a bound for the diameters of  \( \Pi \Lambda_L^{(N)}(\bx) \) and \( \Pi \Lambda_L^{(N)}(\by) \). Since \( \Lambda_L^{(N)}(\bx) \) is fully interactive, the union
\begin{equation}\label{eq:connected_cubes}
\bigcup_{j=1}^N \Big\{ u \in \bZ^d \colon \| u - x_j \|_\infty \leq L + C_\cU \Big\}
\end{equation}
is connected. Otherwise, we could choose \( \cS_1 \) as one of the connected components and \( \cS_2 \) as the union over the cubes \( \{ u \in \bZ^d \colon \| u- x_j \|_\infty \leq L \} \) which are not contained in \( \cS_1 \). Since the union in \eqref{eq:connected_cubes} is connected, we obtain that 
\begin{equation*}
\operatorname{diam}(\Pi \Lambda_L(\bx)) \leq 2N(L+C_\cU) < 4NL,
\end{equation*}
by our assumption on $L$. Since \( \Lambda_L(\by) \) is also fully interactive, the diameter of \( \Pi \Lambda_L(\by) \) obeys the same upper bound. Since \( d_S(\bx,\by) \geq 8NL \), there exist \( 1 \leq i , j \leq N \) such that \( \| x_i - y_j \|_\infty \geq 8NL \). For any \( u \in \Pi \Lambda_L(\bx) \) and \( v \in \Pi \Lambda_L(\by) \), it follows that 
\begin{equation}
\| u - v \|_\infty \geq \| x_i - y_j \|_\infty - \operatorname{diam}(\Pi \Lambda_L(\bx))- \operatorname{diam}(\Pi \Lambda_L(\by)) > 8NL - 4NL - 4NL = 0.
\end{equation}
Thus, \( \Pi \Lambda_L(\bx) \) and \( \Pi \Lambda_L(\by) \) are disjoint. 
\end{proof}

\subsection{Decay estimates on localizing cubes}

In this subsection, we prove that an eigenfunction $\psi$ on $\Lambda_L$ with eigenvalue $\mu$ decays over a smaller cube $\Lambda_\ell$ as long as $\mu$ is not too close the spectrum of the restricted Hamiltonian $H_{\Lambda_\ell}$. This forms the basis of an iteration scheme over the cubes in the cover $C_{L,\ell}$. 

In the statement of the next lemma, we take $\widetilde{\tau}=(1+\tau)/2$.

\begin{lemma}\label{lemma:local_good_cube}
Let \( \Theta \subseteq \bZ^{Nd} \) be symmetric and let \( \Lambda_\ell= \Lambda_\ell^{(N)} \subseteq \Theta \) be \( m \)-localizing. Let 
\( \psi \colon \Theta \rightarrow \mathbb{R} \) be a generalized eigenfunction for \( H_\Theta \) with generalized eigenvalue \( \mu \). We assume that 
\begin{equation}\label{eq:condition_spectral_and_l}
\dist(\mu, \sigma(H_{\Lambda_\ell})) \geq \frac{1}{2} e^{-L^\beta} \quad \text{and} \quad 
\ell \geq C(m,\tau,\gamma,\beta,d) (N \log(2+N))^{\frac{1}{\tau}}. 
\end{equation}
Then, we have for \( m^\prime := m \cdot (1- 3 \ell^{- \frac{1-\tau}{2}}) \) and  all \( \by \in \Lambda_\ell^{\Theta,\ell^{\widetilde{\tau}}} \)  that
\begin{equation}\label{eq:local_good_cube}
|\psi(\by)| \leq  \max_{\bv \in \partial^\Theta_\ex \Lambda_\ell} e^{-m^\prime d_S(\by,\bv)} |\psi(\bv)|.
\end{equation}
\end{lemma}

\begin{remark}
This is the \( N \)-particle analog of \cite[Lemma 2.2]{EK19} on the spectral interval \( I = \bR \). Without the spectral projection, however, the proof simplifies significantly. \\
For future use, we remark that the constant $C(m,\tau,\gamma,\beta,d)$ can be chosen as decreasing in $m$. 
\end{remark}

\begin{proof}

We denote by \( \{ (\varphi_\nu,\bx_\nu) \}_{\nu \in \widetilde{\sigma}(H_{\Lambda_\ell})} \) the eigenfunctions and localization centers of \( H_{\Lambda_\ell} \). Recall that  \( \widetilde{\sigma}(H_{\Lambda_\ell}) \) denotes the spectrum of \( H_{\Lambda_\ell} \) with the eigenvalues (possibly) repeated according to their multiplicities. For any \( \by \in \Lambda_\ell^{\Theta,\ell_{\ttau}} \), we have that 
\begin{equation}\label{eq:local_good_cube_1}
\psi(\by) = \sum_{\nu\in \tspec}  \langle \varphi_\nu, \psi \rangle \varphi_\nu(\by). 
\end{equation}
We can estimate the inner product by 
\begin{align*}
|\langle \varphi_\nu, \psi\rangle| &= |\mu-\nu|^{-1} |\langle (H_\Theta-\nu) \varphi_\nu, \psi \rangle| \\
&= |\mu-\nu|^{-1} |\langle \Gamma_{\partial^\Theta \Lambda_\ell} \varphi_\nu, \psi \rangle| \\
&\leq 2 e^{ L^\beta} \sum_{(\bu,\bv)\in \partial^\Theta \Lambda_\ell} |\varphi_\nu(\bu)| |\psi(\bv)|. 
\end{align*}
After inserting this into \eqref{eq:local_good_cube_1}, we obtain that
\begin{align*}
|\psi(\by)| &\leq 2 e^{ L^\beta} \sum_{\nu \in \tspec}   \sum_{(\bu,\bv)\in \partial^\Theta \Lambda_\ell} |\varphi_\nu(\bu)| |\varphi_\nu(\by)| |\psi(\bv)| \\
&\leq 2 e^{L^\beta} \sum_{\nu \in \tspec}   \sum_{(\bu,\bv)\in \partial^\Theta \Lambda_\ell} e^{m(\ell^\tau-d_S(\bu,\bx_\nu))} e^{m(\ell^\tau-d_S(\by,\bx_\nu))} |\psi(\bv)| \\
&\leq 2  e^{L^\beta} e^{2m \ell^\tau} \sum_{\nu \in \tspec}   \sum_{(\bu,\bv)\in \partial^\Theta \Lambda_\ell}e^{-m d_S(\bu,\by)} |\psi(\bv)| \\
&\leq  2 |\Lambda_\ell| |\partial^\Theta \Lambda_\ell| e^{L^\beta} e^{m (2\ell^\tau+1)}   \max_{\bv \in \partial^\Theta_{\ex}  \Lambda_\ell} e^{-m d_S(\by,\bv)} |\psi(\bv)| \\
&\leq 4Nd (N! (2\ell+1)^{Nd})^2 e^{L^\beta} e^{m (2\ell^\tau+1)} e^{-(m-m^\prime) \ell^\ttau} \cdot 
\max_{\bv \in \partial^\Theta_{\ex} \Lambda_\ell} e^{-m^\prime d_S(\by,\bv)} |\psi(\bv)|. 
\end{align*}
In the last line, we used the condition $\by \in \Lambda_\ell^{\Theta,\ell^{\widetilde{\tau}}}$. In order to complete the argument, it remains to prove that 
\begin{equation}\label{eq:local_good_cube_2}
 4Nd (N! (2\ell+1)^{Nd})^2 e^{L^\beta} e^{m (2\ell^\tau+1)} e^{-(m-m^\prime) \ell^\ttau} \leq 1. 
\end{equation}
Using our choice of \( m^\prime \), we have that 
\begin{equation*}
(m-m^\prime) \ell^{\ttau} \geq 3 m \ell^{\ttau} \ell^{-\frac{1-\tau}{2}} = 3m \ell^\tau. 
\end{equation*}
Thus, \eqref{eq:local_good_cube_2} reduces to 
\begin{equation*}
\ell^\tau \geq 1  +  m^{-1} \big(  \ell^{\gamma \beta} + 2Nd \log(2l+1) + 2 \log(N!) + \log(4Nd) \big).  
\end{equation*}
This is implied the the condition $\tau>\gamma \beta$ on the parameters and the lower bound on $\ell$. 
\end{proof}

\subsection{Buffered cubes}

When proving that a large cube $\Lambda_L$ is $m$-localizing (or $m^\prime$-localizing with $m^\prime$ close to $m$), we would ideally like all smaller cubes $\Lambda_\ell(\ba)$ in the cover $C_{L,\ell}$ to be $m$-localizing. Unfortunately, since we want the probability of $m$-localization to increase in the sidelength of the cube, this is not possible. This problem already occurs in the one-particle setting and we refer to the lecture notes of Kirsch \cite[Section 9]{Kirsch07} for a more detailed explanation. In this subsection, which is entirely analytic, we control the influence of small ``bad" regions on a larger scales. The main idea is that the bad region should be surrounded by a buffer of good cubes. 

\begin{definition}[Buffered cubes]\label{definition:buffered}
We call a symmetric set $\Ups \subseteq \Lambda_L \subseteq \bZ^{Nd}$ a buffered cube in $\Lambda_L$ if the following holds:
\begin{enumerate}[(i)]
\item \( \Ups \) is of the form 
\begin{equation*}
\Ups =  \Lambda_{R}(\bb) \medcap \Lambda_L,
\end{equation*}
with $\bb\in \Lambda_L$ and $\ell \leq R \leq L$. 
\item There exists a set of good centers \( \cG_\Ups \subseteq \Xi_{L,\ell}\), where  $\Xi_{L,\ell}$ is as in Definition \ref{definition:cover}, such that the cubes
$\{ \Lambda_\ell(\ba) \}_{\ba \in \cG_\Ups} $ are $m$-localizing and
\begin{equation*}
\partial^{\Lambda_L}_{\text{in}} \Ups \subseteq \bigcup_{\ba \in \cG_\Ups} \Lambda_\ell^{\Lambda_L,\ell}(\ba). 
\end{equation*}
\end{enumerate}
\end{definition}

\begin{figure}\label{figure:buffered}
\begin{tikzpicture}[scale=0.8]
\draw[ultra thick] (0,0) rectangle (8,8);
\node at (7,7) {\huge$\Lambda_L$};


\draw[red, ultra thick, fill=red, fill opacity=0.4] (2.5,1.5) -- (5.5,1.5)-- (5.5,4.5) -- (4.5,4.5) -- (4.5,5.5)--(1.5,5.5)--(1.5,2.5) -- (2.5,2.5) -- (2.5,1.5);

\draw[OliveGreen,ultra thick,fill=gray, fill opacity=0.05] (2,2) rectangle (3,3);
\draw[OliveGreen,ultra thick,fill=gray, fill opacity=0.05] (2,1) rectangle (3,2);
\draw[OliveGreen,ultra thick,fill=gray, fill opacity=0.05] (3,1) rectangle (4,2);
\draw[OliveGreen,ultra thick,fill=gray, fill opacity=0.05] (4,1) rectangle (5,2);
\draw[OliveGreen,ultra thick,fill=gray, fill opacity=0.05] (5,1) rectangle (6,2);

\draw[OliveGreen,ultra thick,fill=OliveGreen, fill opacity=0.3] (5,2) rectangle (6,3); 

\draw[OliveGreen,ultra thick,fill=gray, fill opacity=0.05] (5,3) rectangle (6,4);
\draw[OliveGreen,ultra thick,fill=gray, fill opacity=0.05] (5,4) rectangle (6,5);
\draw[OliveGreen,ultra thick,fill=gray, fill opacity=0.05] (4,4) rectangle (5,5);

\draw[OliveGreen,ultra thick,fill=gray, fill opacity=0.05] (4,5) rectangle (5,6);
\draw[OliveGreen,ultra thick,fill=gray, fill opacity=0.05] (3,5) rectangle (4,6);
\draw[OliveGreen,ultra thick,fill=OliveGreen, fill opacity=0.3](2,5) rectangle (3,6);                         

\draw[OliveGreen,ultra thick,fill=gray, fill opacity=0.05] (1,5) rectangle (2,6);
\draw[OliveGreen,ultra thick,fill=gray, fill opacity=0.05] (1,4) rectangle (2,5);
\draw[OliveGreen,ultra thick,fill=gray, fill opacity=0.05] (1,3) rectangle (2,4);
\draw[OliveGreen,ultra thick,fill=gray, fill opacity=0.05] (1,2) rectangle (2,3);

\node at (6,6) {\huge \textcolor{red}{ $\Ups$}};
\node at (2.5,6.75) {\Large\textcolor{OliveGreen}{$\Lambda_\ell(a)$}};

\end{tikzpicture}
\caption{\small We display a buffered cube \( \Ups \subseteq \Lambda_L\). The large black box corresponds to \( \Lambda_L \) and the red region corresponds to \( \Ups \). The green boxes illustrate a subset of the localizing boxes \( \Lambda_\ell(a)\), \( a \in \mathcal{G}_\Ups \). The shaded green area displays a single symmetric localizing box. For illustrative purposes, we accepted two differences between this illustration and our argument. By definition, the full family  \(\{ \Lambda_\ell(a)\}_{a\in \mathcal{G}_\Ups} \) has to contain overlapping cubes, but the figure only shows disjoint cubes. While \( \Ups \) accounts for a large portion of \( \Lambda_L \) in the figure, our estimates show that the buffered cubes are small compared to \( \Lambda_L \). }
\end{figure}
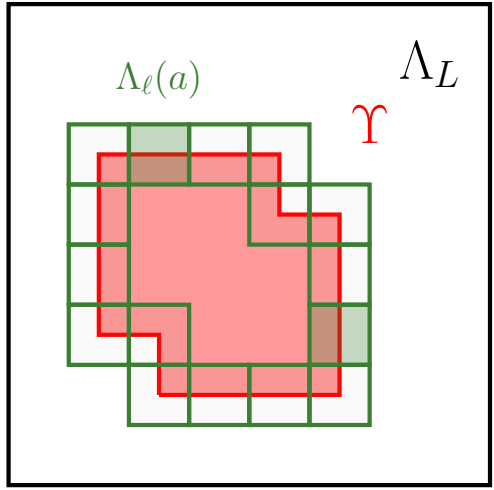

While we do not strictly require that $\Lambda_\ell(\bb)$ in Definition \ref{definition:buffered} is bad, i.e., not $m$-localizing, this will be the case for most of our buffered cubes. 
We illustrate this definition in Figure \ref{figure:buffered}. \\

We now present a short auxiliary lemma that will be used in the proof of Lemma \ref{lemma:buffered} below. The statement and proof are a minor modification of \cite[Lemma 2.1]{EK19}. 

\begin{lemma}\label{lemma:crude_move_to_boundary}
Let \( \Theta \subseteq \bZ^{Nd} \) be symmetric and let \( \psi \) be a generalized eigenfunction of \( H_\Theta \) with generalized eigenvalue \( \mu \). Let \( \Phi \subseteq \Theta \) be finite and symmetric, \( \eta > 0 \), and suppose
\begin{equation}\label{eq:crude_spectral_separation}
\operatorname{dist}(\mu,\sigma(H_\Phi)) \geq \eta. 
\end{equation}
Then, we have the estimate
\begin{equation}
\| \psi \|_{\ell^2(\Phi)} \leq 2Nd \eta^{-1} \Big( \# \partial^\Theta_{\text{ex}} \Phi\Big)^{\frac{1}{2}} \max_{\mathbf{v} \in \partial^\Theta_{\text{ex}} \Phi} |\psi(\mathbf{v})|. 
\end{equation}
\end{lemma}

Except for the weak (and necessary) non-resonance condition \eqref{eq:crude_spectral_separation}, Lemma \ref{lemma:crude_move_to_boundary} requires no information on $\Theta,\psi$, and $\mu$. The lemma never yields any decay by itself, but provides a cheap way to leave a bad region.

\begin{proof}
Since \( (H_\Theta - \mu) \psi = 0 \), it follows from the geometric decomposition \eqref{eq:geometric_decomposition} that 
\begin{equation*}
(H_\phi - \mu ) \chi_{\Phi} \psi = - \chi_{\Phi} \Gamma_{\partial^\Theta \Phi} \psi = -  \chi_{\Phi} \Gamma_{\partial^\Theta \Phi}  (\chi_{\partial^{\Theta}_{\text{ex}} \Phi} \psi).
\end{equation*}
Using this, we obtain 
\begin{align*}
\| \chi_\Phi \psi \|_{\ell^2} 
&= \| (H_\Phi-\mu)^{-1}  \chi_{\Phi} \Gamma_{\partial^\Theta \Phi}  (\chi_{\partial^{\Theta}_{\text{ex}} \Phi} \psi) \|_{\ell^2}  \\
&\leq \| (H_\Phi-\mu)^{-1} \|_{\ell^2\rightarrow \ell^2}  \|  \Gamma_{\partial^\Theta \Phi}  \|_{\ell^2\rightarrow \ell^2} \| \chi_{\partial^{\Theta}_{\text{ex}} \Phi} \psi \|_{\ell^2} \\
&\leq 2Nd \eta^{-1} \Big( \# \partial^\Theta_{\text{ex}} \Phi\Big)^{\frac{1}{2}} \max_{\mathbf{v} \in \partial^\Theta_{\text{ex}} \Phi} |\psi(\mathbf{v})|. \qedhere
\end{align*}
\end{proof}

\begin{lemma}\label{lemma:buffered}
Let \( \Lambda_L = \Lambda_L(\bx_0) \), where \( \bx_0 \in \bZ^{Nd} \), and let \( (\psi,\mu) \) be an eigenpair for \( H_{\Lambda_L} \). Let \( \Ups \subsetneq \Lambda_L \) be an m-buffered cube and suppose that 
\begin{equation*}
\dist(\mu,\sigma(H_\Ups))\geq \frac{1}{2} e^{-L^\beta} \qquad \text{and} \qquad \min_{\ba \in \cG_\Ups} \dist(\mu,\sigma(H_{\Lambda_\ell(\ba)})) \geq \frac{1}{2} e^{-L^\beta}. 
\end{equation*}
Further assume that 
\begin{equation*}
\ell \geq C(m,\tau,\gamma,\beta,d) (N \log(2+N))^{\frac{1}{\tau}}. 
\end{equation*}
Then, we have for \( m^\prime = m( 1- 3\ell^{\frac{1-\tau}{2}}) \) that 
\begin{equation}\label{eq:buffered_decay}
\max_{\by \in \Ups} |\psi(\by)| \leq e^{-\frac{m^\prime}{2} \ell} \max_{\ba \in \cG_\Ups} \max_{\bv \in \partial_{\ex}^{\Lambda_L} \Lambda_\ell(\ba)} |\psi(\bv)|. 
\end{equation}
\end{lemma}

This is the \( N\)-particle analog of \cite[Lemma 2.4]{EK19}. The precise form of the decay in \eqref{eq:buffered_decay} is inessential. It is only important that the pre-factor is bounded by $1/2$, say, which prevents repeated (or infinitely many) visits to the buffered cubes.

\begin{proof}
From Lemma \ref{lemma:crude_move_to_boundary},  we obtain that 
\begin{equation*}
\max_{\by \in \Ups} |\psi(\by)| \leq 4Nd~  e^{L^\beta} |\partial^{\Lambda_L}_{\ex} \Ups|^{\frac{1}{2}} \max_{\bu \in \partial^{\Lambda_L}_\ex \Ups} |\psi(\bu)| \leq 4Nd  e^{L^\beta} |\Lambda_L|^{\frac{1}{2}} \max_{\bu \in \partial^{\Lambda_L}_\ex \Ups} |\psi(\bu)|. 
\end{equation*}
From the definition of a buffered cube, we have that \( \bu \in \Lambda_\ell^{\Lambda_L,\ell}(\ba_\bu) \) for some \( \ba_\bu \in \cG_\Ups \). From the lower bound on $\ell$ and Lemma~\ref{lemma:local_good_cube}, it follows that 
\begin{equation*}
\max_{\by \in \Ups} |\psi(\by)| \leq  4Nd e^{L^\beta} |\Lambda_L|^{\frac{1}{2}} e^{-m^\prime \ell }  \max_{\bv \in \partial_{\ex}^{\Lambda_L} \Lambda_\ell(\ba)} |\psi(\bv)|.
\end{equation*}
It therefore remains to prove that 
\begin{equation*}
4Nd e^{L^\beta} |\Lambda_L|^{\frac{1}{2}} e^{-\frac{m^\prime}{2} \ell}\leq 1. 
\end{equation*}
Using the definition of $m^\prime$ and that $\ell\geq 6^{\frac{2}{1-\tau}}$, we obtain $m^\prime \geq m/2$.  Thus, it suffices to show 
that
\begin{equation*}
\frac{m}{4} \ell \geq \ell^{\gamma \beta} + \frac{1}{2} \log(N!) + \frac{Nd}{2} \log(2\ell^\gamma+1)+\log(4Nd). 
\end{equation*}
Since $\gamma \beta<1$, this follows from the lower bound on $\ell$. 
\end{proof}

\section{Initial Step}

In this section, we proceed with the initial step of the eigenfunction multi-scale analysis, and in particular, we establish a localization estimate for an initial scale, $\ell$. We record the main result of this section in the following.
\begin{proposition}\label{prop:initial_scale}
Let \( \ell \geq 1 \), let \( m > 0\), let $\delta>0$, and assume that the disorder parameter, $\lambda$, satisfies
\begin{equation*}
\lambda \geq 2Nd \| \rho\|_{\infty} (1+e^m) ~(N!)^2~ (2\ell+1)^{2Nd}~ \delta^{-1}. 
\end{equation*}
Then, we have 
\begin{equation}
\inf_{\mathbf{a} \in \bZ^{Nd}} \mathbb{P} \{\Lambda_{\ell}^{(N)}(\mathbf{a})\textup{ is $m$-localizing for $H$}\} \geq 1 - \delta.
\end{equation}
\end{proposition}
\begin{remark}
In our application of this proposition, we will in fact take $\delta$ to be $\ell^{-c}$, for some exponent $c \equiv c(N) >0$, and hence we see from the statement that by taking $\lambda$ sufficiently large, we can ensure $c \gg1$.
\end{remark}

The previous proposition should be compared to \cite[Prop 4.2]{EK16}.  In order to establish Proposition~\ref{prop:initial_scale}, we will prove Lemma~\ref{lemma:initial_scale_analytic}, which should be compared with  \cite[Lemma 4.4]{EK16}. Lemma ~\ref{lemma:initial_scale_analytic} states that under a certain separation condition on the potential, recorded in \eqref{eq:potential_separation} below, we have the required eigenfunction decay. The argument for this lemma will be purely analytic, and not rely on any probabilistic estimates. Proposition~\ref{prop:initial_scale} will then follow from establishing bounds for the probability of the event that the separation condition holds. The difference between the single and multi-particle setting lies in the permutation invariance of the Hamiltonian, which leads to a degenerate spectrum. 

We first present the analytic portion of the initial scale estimate. This part of the argument is based on Gershgorin's disc theorem, which we now recall.

\begin{lemma}[Gershgorin's disc theorem]
Let \( A = (a_{jk})_{j,k=1}^n  \in \mathbb{C}^{n\times n} \) be a complex matrix. For all \( 1 \leq j \leq n \), we define \( R_j := \sum_{1\leq k \leq n \colon k\neq j} |a_{jk}| \). Then, the eigenvalues of \( A \) are contained in 
\begin{equation*}
\bigcup_{j=1}^n \{ z \in \bC^n \colon |z-a_{jj}|\leq R_j \}. 
\end{equation*}
\end{lemma}

We now establish the analytic lemma. To simplify the notation, we write $\bx \not \sim_\pi \by$ for any $\bx,\by \in \bZ^{Nd}$ satisfying $\bx \neq \pi \by$ for all $\pi \in S_N$. 
\begin{lemma}\label{lemma:initial_scale_analytic}
Let \( \Theta \subseteq \bZ^{Nd} \)  be finite and symmetric, let  \( \eta > 4Nd \), and assume for any \( \bx, \by \in \Theta \) satisfying \( \bx \not \sim_\pi \by \) that 
\begin{equation}\label{eq:potential_separation}
|(\lambda V(\bx)+ U(\bx)) -( \lambda V(\by)+ U(\by)) | \geq \eta. 
\end{equation}
Then, for all normalized eigenfunctions \( \varphi \in \ell^2(\Theta) \) there exists an \( \bx \in \Theta \) such that 
\begin{equation}\label{eq:initial_scale_decay}
|\varphi(y)|\leq \Big( \frac{2Nd}{\eta-2Nd}\Big)^{\min_{\pi\in S_N} \| \by - \pi \bx \|_1} \qquad \text{for all } \by \in \Theta. 
\end{equation}
\end{lemma}
\begin{proof}
For any \( x\in \Theta \), we define the Gershgorin disc \( \mathbb{D}_\bx \) by 
\begin{equation}
\bDx := \{ z \in \mathbb{C} \colon |z- (\lambda V(\bx) +  U(\bx))| \leq 2Nd \}. 
\end{equation}
From Gershgorin's disc theorem, we obtain that
\begin{equation}\label{eq:gershgorin}
\sigma( H_\Theta) \subseteq \bigcup_{x\in \Theta} \bDx. 
\end{equation}
Let \( \varphi \in \ell^2(\Theta) \) be a normalized eigenfunction of \( H_\Theta \) with eigenvalue \( \mu \). From \eqref{eq:gershgorin}, we see that \( \mu \in \bDx \) for some \( x \in \Theta \). From the assumption \eqref{eq:potential_separation}, it follows for all \( y \in \Theta \) satisfying \( y \not \sim_\pi x \) that 
\begin{equation*}
|\mu - (\lambda V(\by) +  U(\by))| \geq |(\lambda V(\bx)+ U(\bx)) -( \lambda V(\by)+ U(\by)) | - 2Nd \geq \eta- 2Nd. 
\end{equation*}
We obtain for all \( \by \in \Theta \) satisfying \( \by \not \sim_\pi \bx \) that 
\begin{equation}\label{eq:initial_scale_single_step}
\begin{aligned}
|\varphi (\by)|&= |\langle \delta_{\by}, \varphi \rangle| \\
&= | ( \mu - \lambda V(\by) -  U(\by))|^{-1} |\langle (H_\Theta - \lambda V(\by) -  U(\by)) \delta_\by , \varphi \rangle | \\
& \leq (\eta-2Nd)^{-1} \sum_{ \substack{\bz  \in \Theta\colon \\ \| \bz - \by\|_1 = 1  }} |\varphi(\bz)|  \\
&\leq \frac{2Nd}{\eta-2Nd} \max_{\substack{\bz \in \Theta\colon \\ \| \bz - \by \|_1 =1 }} |\varphi(\bz)|. 
\end{aligned}
\end{equation}
We obtain \eqref{eq:initial_scale_decay} by iterating \eqref{eq:initial_scale_single_step} until we reach the orbit \( S_N \bx = \{ \pi \bx \colon \pi \in S_N \}\). 
\end{proof}

Having established Lemma~\ref{lemma:initial_scale_analytic}, all that remains is establishing the aforementioned probabilistic bounds, which we turn to now.

\begin{proof}[Proof of Proposition \ref{prop:initial_scale}:] 
Let \( \eta := (1+e^m) 2Nd> 4Nd \). Using Lemma \ref{lemma:initial_scale_analytic}, it remains to prove that the separation condition \eqref{eq:potential_separation} is satisfied with probability at least \( 1- \delta\). We recall from \eqref{eq:number_operator} that the number operator \( N_u \colon \bZ^{Nd}\rightarrow \mathbb{N} \), where $u \in \bZ^d$, is given by 
\begin{equation}
N_u(\bx) = \# \{ 1 \leq j \leq N \colon x_j = u \}. 
\end{equation}
Using this notation, we can rewrite the random potential \( V(\bx) \) as 
\begin{equation}
V(\bx) = \sum_{u \in \bZ^d} \cV(u) N_u(\bx). 
\end{equation}
For any pair of particle-configurations \( \bx,\by \in \bZ^{Nd} \) satisfying \( x \not \sim_\pi \by \), it holds that \( N_w(\bx) \neq N_w(\by) \) for some \( w \in \bZ^d \). After splitting
\begin{equation*}
\lambda (V(\bx)-V(\by)) +  (U(\bx)-U(\by)) =: \lambda \cV(w) (N_w(\bx)-N_w(\by)) + A,
\end{equation*}
this yields
\begin{align*}
&\bP\Big( |(\lambda V(\bx)+ U(\bx)) -( \lambda V(\by)+ U(\by)) | < \eta \Big) \\
&= \bE \Big[ \bP\Big( |\lambda \cV(w) (N_w(\bx)-N_w(\by)) + A | < \eta \Big| \{ V(u)\}_{u\neq w} \Big) \Big] \\
&\leq \frac{\| \rho\|_\infty \eta}{\lambda (N_w(\bx)-N_w(\by))} \\ 
&\leq \lambda^{-1} \| \rho\|_\infty \eta. 
\end{align*}
By using a union bound, we obtain that
\begin{align*}
&\bP\Big( \exists \bx, \by \in \Lambda_\ell(\mathbf{a})\colon  \bx \not \sim_\pi \by ~ \text{and} ~  |(\lambda V(\bx)+ U(\bx)) -( \lambda V(\by)+ U(\by)) | < \eta \Big) \\
&\leq \lambda^{-1} \| \rho \|_\infty \eta~  (\# \Lambda_\ell(\mathbf{a}))^2 \\
&\leq  \lambda^{-1} 2Nd\| \rho\|_{\infty} (1+e^m)~ (N!)^2(2\ell+1)^{2Nd}. 
\end{align*}
Due to the assumption on \( \lambda\), this yields the desired estimate. 
\end{proof}

\section{Inductive Step}\label{section:inductive}

A complication of the current work compared to the single particle setting of \cite{EK16} is that in addition to inducting on the scales, we need to induct on the particle number. We first address the latter problem, with an approach which should be compared to \cite{CS09b}. We will use the following proposition to treat the case of partially interactive boxes, which will handle one of the scenarios we need to address in the proof of Theorem \ref{thm:ind_scale} below.

\begin{proposition}[Induction step in the particle number \( n \)]\label{proposition:induction_on_n}
Let \( N \geq 2 \), let \( (p(n))_{n=1,\hdots,N-1} \subseteq [1,\infty) \), let \( L \geq 2 \), and assume for all \( 1 \leq n \leq N-1 \) that
\begin{equation}
\inf_{\bx_0 \in \bZ^{nd}} \bP\Big( \Lambda_L^{(n)}(\bx_0) \text{ is } m\text{-localizing}\Big) \geq 1 - L^{-p(n)}. 
\end{equation}
If \( \Lambda_L^{(N)}(\bx) \) is partially interactive, we have that 
\begin{equation*}
 \bP\Big( \Lambda_L^{(N)}(\bx) \text{ is } m\text{-localizing}\Big) \geq 1 - L^{-\widetilde{p}(N)}, \qquad \text{where} \qquad \widetilde{p}(N):= \min_{1\leq n \leq N-1} p(n) - \frac{\log(2)}{\log(L)}. 
\end{equation*}
\end{proposition}

\begin{remark}\label{rem:exp}
In applications of this proposition, we will choose the disorder parameter $\lambda$ in Proposition \ref{prop:initial_scale} large enough to ensure that $p(n)$, and hence $\tilde{p}(N)$, is sufficiently large.
\end{remark}

\begin{remark}
The idea behind the proof is that a partially interacting \( N \)-particle cube can be split into non-interacting cubes with fewer particles. The \( N \)-particle eigensystem can then be obtained as a tensor product of the eigensystems with fewer particles. Unfortunately, the symmetrization of the cube makes the implementation of this idea slightly cumbersome. 
\end{remark}

\begin{proof}
Let \( N_1, N_2, \cS_1 \), and \( \cS_2 \) be as in Definition \ref{definition:partially_interacting}. Using the permutation invariance of the Hamiltonian \( H^{(N)} \) and the cube \( \Lambda_L^{(N)}(\bx) \), we may assume that \( x_1,\hdots,x_{N_1} \in \cS_1 \) and \( x_{N_1+1},\hdots,x_{N} \in \cS_2 \). For any \( \cJ \subseteq \{ 1, \hdots , N\} \) with \( \# \cJ = N_1 \), we define
\begin{equation}\label{eq:spatial_tensor_product}
(1_{\cS_1} \otimes_\cJ 1_{\cS_2} ) (\by) := \prod_{j\in \cJ} 1_{\cS_1}(y_j) \cdot \prod_{j\in \cJ^c} 1_{\cS_2}(y_j). 
\end{equation}
We note that different sets \( \cJ \) lead to disjoint supports in \eqref{eq:spatial_tensor_product}. Since $\Lambda_L^{(N)}(\bx)$ is partially interactive, it follows from Definition \ref{definition:partially_interacting} that 
\begin{equation*}
1= \sum_{\cJ} (1_{\cS_1} \otimes_\cJ 1_{\cS_2} ) (\by). 
\end{equation*}
This leads to the orthogonal decomposition
\begin{equation}
\ell^2(\Lambda_L(\bx)) = \bigoplus_{\cJ} \ell^2_\cJ(\Lambda_L(\bx)) := \bigoplus_{\cJ} \Big\{  \varphi \in \ell^2(\Lambda_L(\bx))\colon (1_{\cS_1} \otimes_\cJ 1_{\cS_2} ) \varphi = \varphi \Big \}.  
\end{equation}
The subspace \( \ell_\cJ^2(\boxx) \) contains wave-functions \( \varphi \) supported on particle configurations \( \by \in \boxx \) with \( y_j \in \cS_1 \) for all \( j \in \cJ \) and \( y_j \in \cS_2 \) for all \( j \in \cJ^c \). The operator \( H^{(N)}_{\boxx} \) leaves each subspace \(  \ell^2_\cJ(\Lambda_L(\bx)) \) invariant and we can decompose
\begin{equation}\label{eq:decomposition_HN}
H^{(N)}_{\boxx} = \bigoplus_{\cJ} H^{(\cJ)}_{\boxx}, \quad \text{where} \quad H^{(\cJ)}_{\boxx} = H^{(N)}_{\boxx} \Big|_{\ell_\cJ^2(\boxx)}. 
\end{equation}
The eigensystems of \( H^{(N)}_{\boxx} \) can then be obtained as a union of the eigensystems of \( H^{(\cJ)}_{\boxx} \). 

Let now \( \cJ_0:= \{ 1, \hdots , N_1\} \), and for each \( \cJ \), we fix a permutation \( \pi_{\cJ} \in S_N \) satisfying \( \pi_\cJ(\cJ) = \cJ_0 \). With this notation, we have that 
\begin{equation*}
\pi_\cJ \ell^2_\cJ(\boxx) = \ell^2_{\cJ_0}(\boxx) \qquad \text{and} \qquad \pi_\cJ \circ H^{(\cJ)}_{\boxx} \circ \pi_\cJ = H^{(\cJ_0)}_{\boxx}.
\end{equation*}
Thus, 
\[
 H^{(\cJ)}_{\boxx}  \textup{ is \(m\)-localizing} \quad \Longleftrightarrow \quad   H^{(\cJ_0)}_{\boxx} \textup{ is \( m \)-localizing}
 \]
Together with the decomposition \eqref{eq:decomposition_HN}, it follows that \( H^{(N)}_{\boxx} \) is \(m\)-localizing if and only if \( H^{(\cJ_0)}_{\boxx} \) is localizing. Consequently, recalling the definition of \( m\)-localizing for operators and boxes, see Definition \ref{definition:localizing}, it follows that
\begin{align}\label{eq:probabilistic_reduction_1}
 \bP\Big( \Lambda_L^{(N)}(\bx) \text{ is } m\text{-localizing}\Big)  &=  \bP\Big( H^{(N)}_{\boxx} \text{ is } m\text{-localizing}\Big)  \\
 &\geq  \bP\Big( H^{(\cJ_0)}_{\boxx} \text{ is } m\text{-localizing}\Big).  \nonumber
\end{align}
We will now establish bounds on the latter probability. Since \( \cJ_0 = \{ 1,\hdots,N_1\} \), we have from the Definition \ref{definition:partially_interacting} that 
\begin{equation*}
\ell_{\cJ_0}^2(\Lambda^{(N)}_{L}(\bx)) = \ell^2(\Lambda^{(N_1)}_L(\bx_{\cJ_0}) ) \otimes \ell^2(\Lambda^{(N_2)}_L(\bx_{\cJ_0^c}) )
\end{equation*}
and 
\begin{equation*}
H^{(\cJ_0)}_{\Lambda_L^{(N)}(\bx)} = \Big( H^{(N_1)}_{\Lambda_L^{(N_1)}(\bx_{\cJ_0})} \otimes \operatorname{Id}\Big) + \Big( \operatorname{Id} \otimes  H^{(N_2)}_{\Lambda_L^{(N_2)}(\bx_{\cJ_0^c})} \Big). 
\end{equation*}
Since the eigensystem for \( H^{(\cJ_0)}_{\Lambda_L^{(N)}(\bx)} \) can be written as a tensor product of the eigensystems of the two operators on the right, we can estimate
\begin{equation}\label{eq:probabilistic_reduction_2}
 \bP\Big( H^{(\cJ_0)}_{\boxx} \text{ is } m\text{-localizing}\Big)  \geq \bP\Big( H^{(N_1)}_{\Lambda_L^{(N_1)}(\bx_{\cJ_0})} ~~ \text{and} ~~ H^{(N_2)}_{\Lambda_L^{(N_2)}(\bx_{\cJ_0^c})} ~ \text{are } m\text{-localizing}\Big).
\end{equation}
Thus by combining \eqref{eq:probabilistic_reduction_1}, \eqref{eq:probabilistic_reduction_2}, and the induction hypothesis, we obtain that
\begin{equation*}
 \bP\Big( \Lambda_L^{(N)}(\bx) \text{ is } m\text{-localizing}\Big) \geq 1- L^{-p(N_1)}- L^{-p(N_2)},
\end{equation*}
which yields the desired estimate. 
\end{proof}

The following is the main inductive step, and constitutes the bulk of the work of this section.

\begin{theorem}[Induction step in the scale \( \ell \)] \label{thm:ind_scale}
Let $d=1$ and let \( (p(n))_{n=1,\hdots,N}\subseteq [1,\infty) \) be a decreasing sequence satisfying
\begin{align}\label{equ:pN_cond}
\Big( \frac{2}{\gamma}-1\Big)^{-1} (4Nd + 2)\leq p(N) \quad \text{and} \quad p(n) \leq \frac{1}{\gamma} \big( p(n-1) - 1 \big) - 2nd- 2 ~~ \text{ for all } 2\leq n \leq N. 
\end{align}
Let
\[
\ell_0 := \ell_0(N, (p(n))_{n=1,\hdots,N} ) = C(\rho,m,\tau,\gamma,\beta,d) \max_{n=1,\hdots,N} \Big( n  + p(n) \Big)^{\frac{1}{\gamma\beta}+\frac{2}{1-\tau \gamma}},
\]
and suppose that for some scale $\ell \geq \ell_0$ and all \( 1 \leq n \leq N \) that
\begin{align}\label{equ:base}
\inf_{\bx_0 \in \bZ^{nd}} \mathbb{P} \{\Lambda_{\ell}^{(n)}(\bx_0)\textup{ is $m$-localizing for $H$}\} \geq 1 - \ell^{-p(n)}  ,
\end{align}
Then, we have for $L = \ell^\gamma$ and all $1\leq n \leq N$ that 
\begin{align}\label{equ:induct}
\inf_{\bx_0 \in \bZ^{nd}} \mathbb{P} \{\Lambda_{L}^{(n)}(\bx_0)\textup{ is $M$-localizing for $H$}\} \geq 1 - L^{-p(n)},
\end{align}
where $M :=  m ( 1- 3\ell^{- \frac{1-\tau}{2}}) (1-250 N^2 \ell^{1-\tau\gamma})$. 
\end{theorem}

\begin{remark}
We emphasize that in \eqref{equ:base}, $n$ is allowed to be $N$ because we are currently inducting on the length scale and not the particle number. We also mention that the constant $C=C(\rho,m,\tau,\gamma,\beta,d)$ can be chosen as decreasing in $m$.
\end{remark}

\begin{proof}
We only prove \eqref{equ:induct} for $n=N$. After minor notational changes, the same argument also applies for $1\leq n < N$.

Suppose \eqref{equ:base} holds for some scale $\ell$. Fix $\bx_0 \in \bZ^{Nd}$ and consider $\Lambda_L(\bx_0)$ for $L$ as in the statement of the theorem. We will prove that \eqref{equ:induct} holds in five steps. Throughout, we let \( C_{L,\ell} \) be an \(\ell\)-cover of \( \Lambda_L(\bx_0) \) as in Definition \ref{definition:cover} and let $\Xi_{L,\ell}$ be the corresponding centers.

\subsection*{Step 1: Definition of the good event \( \cE \)}  The good event \( \cE \) will lead to localization on \( \Lambda_L \) and will satisfy the probabilistic bound \( \bP(\cE)\geq 1 - L^{-p(N)} \) for some $p(N)$. We define \( \cE:= \cE_{\text{PI}} \medcap \cE_{\FI} \medcap \cE_{\NR} \), where the three events contain conditions on partially interactive cubes, conditions on fully interactive cubes, and a non-resonance condition, respectively.

We first define the good event for partially interactive cubes by 
\begin{equation}
\mathcal{E}_{\text{PI}} := \Big\{ \text{all partially interactive cubes in } C_{L,\ell} \text{ are } m\text{-localizing} \Big\}. 
\end{equation} 
Second, we define the good event for fully interactive cubes by 
\begin{equation}
\begin{aligned}
\cE_{\FI} :=& \Big \{ \text{For all } \ba_1,\ba_2 \in \Xi_{L,\ell} \text{ satisfying } d_S(\ba_1,\ba_2) \geq 8N\ell \text{ and such that } \Lambda_\ell(\ba_1) \text{ and } \Lambda_\ell(\ba_2)\\
& \text{ are fully interactive, either } \Lambda_\ell(\ba_1) \text{ or } \Lambda_\ell(\ba_2) \text{ is m-localizing}\Big\}. 
\end{aligned}
\end{equation}
Finally, we define the good event regarding non-resonances by 
\begin{equation}
\begin{aligned}
\cE_{\NR} := \Big\{& \text{For all } \ba_1,\ba_2 \in \Xi_{L,\ell} \text{ satisfying } d_S(\ba_1,\ba_2) \geq 200 N^2 \ell, \\
&\Theta_1= \Lambda_\ell(\ba_1)  \text{ or } \Theta_1=\Lambda_{10N\ell}(\ba_1) \medcap \Lambda_L, \text{and }\Theta_2= \Lambda_\ell(\ba_2)  \text{ or } \Theta_2=\Lambda_{10N\ell}(\ba_2) \medcap \Lambda_L, \\
&\text{it holds that } \dist( \sigma(H_{\Theta_1}),\sigma(H_{\Theta_2})) \geq 1/2 \cdot e^{-L^{-\beta}} \Big\}. 
\end{aligned}
\end{equation}
\subsection*{Step 2: Estimate of the probability of the bad set} ~
In this step, we separately estimate the probabilities of \( \cE_{\PI},\cE_{\FI}, \) and \( \cE_{\NR} \).

We first estimate the probability of \( \cE_{\PI} \). We mention that in the single-particle setting, i.e., $N=1$, there are no partially interactive cubes, and hence $\bP(\cE_{\PI}^c) =0$ holds trivially.  If $N\geq 2$, it follows from Lemma \ref{lemma:cover} and Proposition \ref{proposition:induction_on_n} that 
\begin{align*}
\bP(\cE_{\PI}^c) \leq (2L+1)^{Nd} \ell^{-\widetilde{p}(N)} \leq L^{2Nd - \frac{\widetilde{p}(N)}{\gamma}} \leq \frac{1}{3} L^{-p(N)},
\end{align*}
provided that
\begin{equation}\label{eq:p_cond1}
\widetilde{p}(N) \geq \gamma (p(N)+2Nd+2). 
\end{equation}

We now estimate the probability of \( \cE_{\FI} \). Let \( \ba_1,\ba_2 \in \Xi_{L,l} \) be as in the definition of \( \cE_{\FI} \), that is, such that \( d_S(\ba_1,\ba_2) \geq 8N\ell \) and such that \( \Lambda_\ell(\ba_1) \) and \( \Lambda_\ell(\ba_2) \) are fully interactive. 
By Lemma \ref{lemma:fully_interactive}, it follows that \( \Pi \Lambda_\ell(\ba_1) \) and \( \Pi \Lambda_\ell(\ba_2) \) are disjoint. As a result, the random operators \( H_{\Lambda_\ell(\ba_1)} \) and \( H_{\Lambda_\ell(\ba_2)} \) are probabilistically independent. From a union bound and the induction hypothesis \eqref{equ:base}, we obtain that 
\begin{align*}
\bP(\cE_{\FI}^c) &\leq \sum_{\substack{\ba_1,\ba_2 \in \Xi_{L,l}\colon \\
\Pi \Lambda_\ell (\ba_1), \Pi \Lambda_\ell(\ba_2) \\ \text{are disjoint} }} \bP\big( \{ \text{Both } \Lambda_\ell(\ba_1) \text{ and } \Lambda_\ell(\ba_2) \text{ are not } m\text{-localizing}\} \big) \\
&\leq \sum_{\substack{\ba_1,\ba_2 \in \Xi_{L,l}}} \bP\big( \{ \Lambda_\ell(\ba_1) \text{ is  not } m\text{-localizing}\}\big) \cdot \bP\big( \{ \Lambda_\ell(\ba_2) \text{ is  not } m\text{-localizing}\}\big) \\
&\leq (\# \Xi_{L,\ell})^2 \ell^{-2p(N)} \leq L^{4Nd - \frac{2}{\gamma} p(N)}  \leq \frac{1}{4} L^{-p(N)},
\end{align*}
provided that 
\begin{equation}\label{eq:p_cond2}
p(N) \geq \Big( \frac{2}{\gamma}-1\Big)^{-1} (4Nd + 2). 
\end{equation}
As mentioned in Remark \ref{rem:exp}, the disorder parameter $\lambda$ can be chosen so as to ensure that we can make $\tilde{p}(N)$ arbitrarily large, 
 and hence $p(N)$ can be guaranteed to satisfy the two conditions \eqref{eq:p_cond1} and \eqref{eq:p_cond2} simultaneously. 

It remains to estimate the probability of \( \cE_{\NR} \). Using the multi-particle Wegner estimate, i.e., Proposition \ref{prop:wegner}, we have that 
\begin{align*}
\bP( \cE_{\NR}^c) \leq C(\rho)^{Nd} (2 \# \Xi_{L,\ell})^{2} N^{12 Nd} \ell^{6Nd} e^{-\frac{1}{2} L^\beta} \leq C(\rho)^{Nd} N^{12Nd} L^{8Nd} e^{-\frac{1}{2} L^\beta} \leq \frac{1}{3} L^{-p(N)},
\end{align*}
provided that 
\begin{equation}\label{eq:prob_cond}
\frac{1}{2} L^{\beta} \geq p(N) \log(L) + Nd \Big( \log(C(\rho)) + 12 \log(N) + 8 \log(L)\Big) + \log(3). 
\end{equation}
The last condition \eqref{eq:prob_cond} follows from our lower bound on $\ell$. 

\subsection*{Step 3: Buffered cubes, good cubes, and spectral separation}
For the rest of this proof, we only work on the good event \( \cE \). There exists a \( \bb \in \Xi_{L,\ell} \) such that for all \( \ba \in \Xi_{L,\ell} \) satisfying \( d_S(\ba,\bb) \geq 8N\ell \), the cube \( \Lambda_\ell(\ba) \) is \( m\)-localizing. We define the set of centers for good cubes by 
\begin{equation}
\cG := \{ \ba \in \Xi_{L,\ell}\colon d_S(\ba,\bb)\geq 8N\ell\}.
\end{equation}
We also define the buffered subset \( \Ups \) and the associated set of good cubes  \( \cG_\Ups \) in the buffer by 
\begin{equation}
 \quad \Ups := \Lambda_{10N\ell}(\bb) \medcap \Lambda_L, \quad \text{and} \quad \cG_{\Ups} = \{ \ba \in \Xi_{L,\ell}\colon 8 N \ell \leq d_S(\ba,\bb) \leq 12 N\ell \}. 
\end{equation}
We now verify that \( \Ups \) and \( \cG_{\Ups} \) satisfy the conditions in Definition \ref{definition:buffered}. The property (i) is immediate. To prove (ii), let \( \by \in \partial^{\Lambda_L}_{\text{in}} \Ups \). By the covering lemma  \ref{lemma:cover}, there exists an \( \ba \in \Xi_{L,\ell} \) such that 
\begin{equation*}
\by \in \Lambda_\ell^{\Lambda_L,\ell} (\ba) \subseteq \Lambda_\ell(\ba). 
\end{equation*}
Since \( d_S(\by,\bb) = 10N\ell \), we obtain that \( 9N\ell \leq d_S(\ba,\bb) \leq 11 N\ell \), which implies \( \ba \in \cG_\Ups\).  \\

From the covering lemma \ref{lemma:cover}, we also obtain that 
\begin{equation}\label{step:eq_covering}
\Lambda_L = \Big( \bigcup_{\ba \in \cG} \Lambda_\ell^{\Lambda_L,\ell}(\ba) \Big) \medcup \Ups. 
\end{equation}
We then define 
\begin{equation}
\Cgb  := \{ \Lambda_\ell(\ba) \}_{a\in \cG} \medcup \{ \Ups \}. 
\end{equation}
We view \( \Cgb \) as a modification of the cover \( C_{L,\ell} \), which retains most (but not necessarily all) of the cubes and adds the buffered cube \( \Ups \). We will refer to elements of \( \Cgb \) by \( \Theta \), which allows us to uniformly treat \( \Lambda_\ell(\ba) \) and \( \Ups \) in some arguments below. From the definition of \( \cE_{\NR} \supseteq \cE \), it follows that \( \Cgb \) is spectrally \( L \)-separated, see Definition \ref{definition:spectral_separation}.
 In contrast to \cite{EK19}, we only require a single buffered cube \( \Ups \). This is an advantage of working with polynomial instead of exponential tails in the probabilistic estimates. 
\subsection*{Step 4: Proximity of the eigenvalues at the scales $L$ and $\ell$}
We let \( \mu \) be an eigenvalue of \( H_{\Lambda_L} \). In this step, we show that there exists a \( \Theta_{\mu} \in \Cgb \) such that 
\begin{equation}
\dist( \mu, \sigma(H_{\Theta_{\mu}}) ) < \frac{1}{2} e^{-L^\beta}. 
\end{equation}
Arguing by contradiction, we assume that \( \dist( \mu,\sigma(H_\Theta)) \geq 1/2 \cdot e^{-L^\beta} \) for all \( \Theta \in \Cgb \). We let \( \psi \) be a normalized eigenfunction of \( H_{\Lambda_L} \) with the eigenvalue \(\mu\). Due to the normalization, we have that \( \| \psi\|_{\ell^\infty} \leq \| \psi\|_{\ell^2} = 1\). We claim that this can be upgraded to 
\begin{equation}\label{step:eq_linf}
\| \psi \|_{\ell^\infty} \leq \exp(-\frac{m'}{2} \ell). 
\end{equation}
Once the claim \eqref{step:eq_linf} has been established, the lower bound on $\ell$ leads to 
\begin{equation*}
 \| \psi\|_{\ell^2}^2 \leq N! (2L+1)^{Nd} \exp(-\frac{m'}{2} \ell) < 1,
 \end{equation*}
which is a contradiction. 

To see the claim \eqref{step:eq_linf}, we let \( \by \in \Lambda_L \) be arbitrary. Using the covering property \eqref{step:eq_covering}, it follows that either \( \by \in \Lambda_\ell^{\Lambda_L,\ell}(\ba) \) for some \( \ba \in \cG \) or that \( \by \in \Ups \).  

If \( \by \in \Lambda_\ell^{\Lambda_L,\ell}(\ba) \), it follows from the lower bound on $\ell$, $\tau>\gamma \beta$, Lemma \ref{lemma:local_good_cube}, and the spectral separation \( d(\mu,\sigma(H_{\Lambda_\ell(\ba)}))\geq 1/2 \cdot e^{-L^\beta} \) that 
\begin{equation*}
|\psi(\by)|\leq \max_{\bv \in \partial^{\Lambda_L}_{\ex} \Lambda_\ell(\ba)} e^{-m^\prime d_S(\bv,\by)} |\psi(v)| \leq e^{-m^\prime \ell/2}
\end{equation*}
If \( \by \in \Ups \), we similarly obtain from \( \dist( \mu,\sigma(H_\Ups)) \geq 1/2\cdot e^{-L^\beta} \), \( \dist(\mu,\sigma(H_{\Lambda_\ell(\ba)})) \geq 1/2 \cdot e^{-L^\beta} \) for all \( \ba \in \cG_{\Ups}\subseteq \cG \), and Lemma \ref{lemma:buffered} that
\begin{equation*}
|\psi(\by)|\leq e^{-\frac{m^\prime}{2} \ell} \max_{\ba\in  \cG_{\Ups}} \max_{\bv \in \partial^{\Lambda_L}_\ex \Lambda_\ell(\ba)} |\psi(\bv)| \leq e^{-\frac{m'}{2} \ell}. 
\end{equation*}
This completes the proof of the claim \eqref{step:eq_linf} and hence this step.

\subsection*{Step 5: $M$-localization.}
Similar as in Step 4, we let \( (\psi_\mu,\mu) \) be a normalized eigenpair of \( H_{\Lambda_L} \). It remains to show that there exists a \( \bx_\mu \in \Lambda_L \) such that \( \psi_\mu \) is \( (\bx_\mu,M)\)-localized, i.e., 
\begin{equation}\label{step:eq_localizing}
|\psi_\mu(\by)|\leq e^{-M d_S(\by,\bx_\mu)}
\end{equation}
for all \( \by \in \Lambda_L \) satisfying \( d_S(\by,\bx_\mu) \geq L^\tau \). 
By Step 4, there exists a region \( \Theta_\mu \in \Cgb \) such that \( \dist( \mu,\sigma(H_{\Theta_\mu})) \leq 1/2 \cdot e^{-L^\beta} \). Since \( \Cgb \) is spectrally \( L \)-separated and each set in $\Cgb$ has a diameter of at most $20N\ell$, it follows that 
\begin{equation}\label{step:spectral_separation}
\dist(\mu,\sigma(H_\Theta)) \geq \frac{1}{2} e^{-L^\beta} \qquad \forall \Theta \in \Cgb \text{ s.t. } d_S(\Theta,\Theta_\mu) \geq 160 N^2 \ell. 
\end{equation}
We choose any particle-configuration \( \bx_\mu \in \Theta_\mu \) as our localization center. Now, let \( \by_0 = \by \in \Lambda_L \) satisfy \( d_S(\by_0,\bx_\mu) \geq 200 N^2 \ell\). Due to the lower bound on $\ell$, this assumption is (much) weaker than the assumption $ d_S(\by_0,\bx_\mu) \geq L^\tau$ for \eqref{step:eq_localizing}. By the covering property \eqref{step:eq_covering}, we have \( \by_0 \in \Lambda_\ell^{\Lambda_L,\ell}(\ba) \) for some \( \ba \in \cG \) or \( \by_0 \in \Ups \). We then set \( \Theta =\Lambda_\ell(\ba) \) or \( \Theta= \Ups \), respectively. Since \( \diam_S(\Theta) \leq 20 N \ell \), it follows that 
\begin{equation*}
d_S(\Theta,\Theta_\mu) \geq d_S(\by_0,\Theta_\mu) - \diam_S(\Theta) \geq 200 N^2 \ell - 20 N \ell \geq 160 N^2 \ell. 
\end{equation*}
Thus, it follows from the spectral separation of \( \Cgb \) that \( \dist(\mu,\sigma(H_\Theta)) \geq 1/2 \cdot e^{-L^\beta} \). Next, we apply the decay estimates for good and buffered cubes. 

If \( \Theta= \Lambda_\ell(\ba) \), and hence \( \by_0 \in \Lambda_\ell^{\Lambda_L,\ell}(\ba) \) it follows from Lemma \ref{lemma:local_good_cube} that 
\begin{equation}
|\psi(\by_0)|\leq e^{-m^\prime d_S(\by_0,\by_1)} |\psi(\by_1)|
\end{equation}
for some \( \by_1 \in \partial_\ex^{\Lambda_L} \Lambda_\ell(\ba) \). We call this scenario a good step. In particular, since \( \by_0 \in \Lambda_\ell^{\Lambda_L,\ell}(\ba) \),  it follows that \( d_S( \by_0,\by_1) \geq \ell \) and hence  \( |\psi(\by_0)|\leq 1/2\cdot|\psi(\by_1)| \).

If \( \Theta= \Ups \), and hence \( \by_0 \in \Ups \),  it follows from Lemma \ref{lemma:buffered} that 
\begin{equation}
|\psi(\by_0)| \leq e^{-\frac{m^\prime \ell}{2}} \max_{\ba \in \cG_\Ups} \max_{\bv \in \partial^{\Lambda_L}_{\ex} \Lambda_\ell(\ba)} |\psi(\bv)| \leq \frac{1}{2} |\psi(\by_1)|, 
\end{equation}
where \( \by_1 \in \widetilde{\Ups} := \{ v \in \Lambda_L \colon d_S(\bv,\Ups) \leq 2N\ell \}\). We call this scenario a bad step.

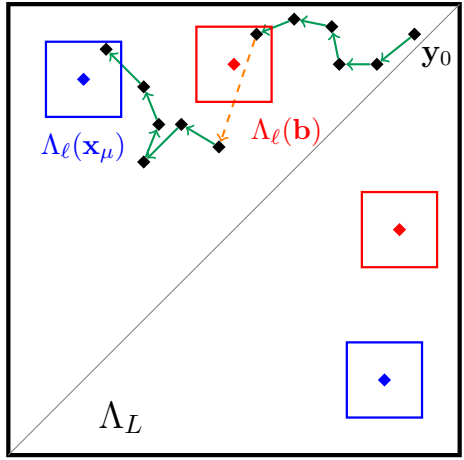
\begin{figure}\label{f:iter}
\tikzmath{\smallL= 0.5; \bigL =3; \oset=0.15; \h = 0.65; }
\begin{tikzpicture}

\draw[ultra thick] (-\bigL, -\bigL) rectangle (\bigL,\bigL);

\draw[gray] (-\bigL, -\bigL)  -- (\bigL, \bigL);


\node at (\bigL-0.3,\bigL-0.7) {\large $\by_0$};

\node(y0)[diamond, fill=black,scale=0.35] at (\bigL-0.6,\bigL-0.4){};
\node(y1)[diamond, fill=black,scale=0.35] at ($(y0)+(-\smallL,-0.4)$){};
\node(y2)[diamond, fill=black,scale=0.35] at ($(y1)+(-\smallL,0)$){};
\node(y3)[diamond, fill=black,scale=0.35] at ($(y2)+(-.1,\smallL)$){};
\node(y4)[diamond, fill=black,scale=0.35] at ($(y3)+(-\smallL,0.1)$){};
\node(y5)[diamond, fill=black,scale=0.35] at ($(y4)+(-\smallL,-0.2)$){}; 

\node(y6)[diamond, fill=black,scale=0.35] at ($(y5)+(-\smallL,-3*\smallL)$){};

\node(y7)[diamond, fill=black,scale=0.35] at ($(y6)+(-\smallL,0.3)$){};
\node(y8)[diamond, fill=black,scale=0.35] at ($(y7)+(-\smallL,-\smallL)$){};
\node(y9)[diamond, fill=black,scale=0.35] at ($(y8)+(0.2,\smallL)$){};
\node(y10)[diamond, fill=black,scale=0.35] at ($(y9)+(-0.2,\smallL)$){};
\node(y11)[diamond, fill=black,scale=0.35] at ($(y10)+(-\smallL,\smallL)$){};

\draw[ForestGreen,thick,->] (y0)--(y1);
\draw[ForestGreen,thick,->] (y1)--(y2);
\draw[ForestGreen,thick,->] (y2)--(y3);
\draw[ForestGreen,thick,->] (y3)--(y4);
\draw[ForestGreen,thick,->] (y4)--(y5);

\draw[ForestGreen,thick,->,orange,dashed] (y5)--(y6);

\draw[ForestGreen,thick,->] (y6)--(y7);
\draw[ForestGreen,thick,->] (y7)--(y8);
\draw[ForestGreen,thick,->] (y8)--(y9);
\draw[ForestGreen,thick,->] (y9)--(y10);
\draw[ForestGreen,thick,->] (y10)--(y11);


\node(bad)[diamond, fill=red,scale=0.35] at (0,\bigL-0.8){};
\redsquare{(bad)}{\smallL};

\node(bad2)[diamond, fill=red,scale=0.35] at (\bigL-0.8,0){};
\redsquare{(bad2)}{\smallL};

\node(res)[diamond,fill=blue,scale=0.35] at (-\bigL+1, \bigL-1) {};
\bluesquare{(res)}{\smallL};
\node(res2)[diamond,fill=blue,scale=0.35] at (\bigL-1, -\bigL+1) {};
\bluesquare{(res2)}{\smallL};

\node at (-\bigL+1.5,-\bigL+0.5) {\Large{$\Lambda_L$}};
\node at ($(res)+(0,-0.9)$) {\textcolor{blue}{$\Lambda_\ell(\mathbf{x}_\mu)$}};
\node at ($(bad)+(0.7,-0.9)$) {\textcolor{red}{$\Lambda_\ell(\mathbf{b})$}};
\end{tikzpicture}
\caption{\small{We illustrate the iteration scheme which generates the sequence $(\by_k)_{k=0}^K$. The bad cube $\Lambda_\ell(\bb)$, which is inside the buffered cube $\Ups$, is displayed in red. The cube $\Lambda_\ell(\bx_\mu)$ is displayed in blue. The good and bad steps in the iteration are displayed in green and orange, respectively.  Certain pre-factors, such as $200N^2$ in \eqref{item:y_iteration_2} below, are ignored in this figure for illustrative purposes. }}
\end{figure}

By iterating this procedure, see Figure \ref{f:iter}, we generate a final iteration index \( K \) and a  sequence of particle configurations \( (\by_k)_{k=0}^K \) satisfying the following properties: 
\begin{enumerate}[(i)]
\item \( \by_0 = \by \), 
\item \label{item:y_iteration_2} \( d_S(\by_k, \bx_\mu) \geq 200 N^2 \ell \) for all \( 0 \leq k \leq K-1 \) and \( d_S(\by_K,\bx_\mu)< 200 N^2 \ell \). 
\item If \( \by_k \not \in \Ups \) and \( 0 \leq k \leq K-1 \), then \( d_S(\by_k,\by_{k+1})\geq \ell \) and 
\begin{equation}\label{step:eq_good}
|\psi(\by_k)|\leq e^{-m^\prime d_S(\by_k,\by_{k+1})} |\psi(\by_{k+1})|.
\end{equation}
We then call the \( k \)-th iteration step good.
\item If \( \by_k \in \Ups \) and \( 0 \leq k \leq K-1 \), then \( \by_{k+1} \in \widetilde{\Ups} \) and 
\begin{equation}\label{step:eq_bad}
|\psi(\by_k)|\leq \frac{1}{2} |\psi(\by_{k+1})|.
\end{equation}
\end{enumerate}
We remark that the final index \( K \) is well-defined, i.e., the iteration stops after finitely many steps, since both \eqref{step:eq_good} and \eqref{step:eq_bad} gain at least a factor of \( 1/2 \).

Since all bad steps occur in \( \widetilde{\Ups} \), we obtain that 
\begin{equation}
\sum_{\substack{0\leq k \leq K-1: \\ k\text{-th step is good}}} d_S(\by_k, \by_{k+1}) \geq d_S(\bx_\mu,\by)- d_S(\bx,\by_K) - \diam(\widetilde{\Ups}) \geq d_S(\bx_\mu,\by) -250 N^2 \ell. 
\end{equation}
Using the a-priori estimate \( |\psi(\by_K)|\leq 1\), we obtain that 
\begin{equation}\label{step:eq_iterated}
|\psi(\by) |\leq \exp\Big( - m^\prime \sum_{\substack{0\leq k \leq K-1: \\ k\text{-th step is good}}} d_S(\by_k, \by_{k+1}) \Big) \leq \exp\Big(- m^\prime (d_S(\bx_\mu,\by) -250 N^2 \ell) \Big). 
\end{equation}
If \( d_S(\by,\bx_\mu) \geq L^\tau \), then \eqref{step:eq_iterated} and our definition of \( M \) implies that
\begin{equation}
|\psi(\by) |\leq \exp( - (m^\prime-M) L^\tau + 250 m^\prime N^2  \ell) \exp(- M d_S(\by,\bx_\mu)) =  \exp(- M d_S(\by,\bx_\mu)) . 
\end{equation}
This completes the proof of \eqref{step:eq_localizing} and hence the proof of the theorem. 
\end{proof}

\section{Proof of main theorem}
We now prove the  main result, Theorem \ref{theorem:AL}.

\begin{proof}
The proof proceeds through an induction on the length scales. We will not induct on the number of particles here, which was previously done in Proposition \ref{proposition:induction_on_n} and the proof of Theorem \ref{thm:ind_scale}. Before we start the induction, however, we need to choose a sequence of decay parameters $(p(n))_{n=1}^{N}$, length scales $(L_k)_{k=0}^\infty$, and (inverse) localization lengths $(m_k)_{k=0}^\infty$. Our choices are motivated by the conditions in Theorem \ref{thm:ind_scale}.

The decay parameters $(p(n))_{n=1}^N$ are defined through a backwards recursion: We set 
\begin{equation*}
p(N) := \max\Big( p, \Big( \frac{2}{\gamma}-1\Big)^{-1} (4Nd+2)\Big) 
\end{equation*}
and for all $2\leq n\leq N$ we define
\begin{equation*}
p(n-1) = \gamma \big( p(n) + 2n d +2 \big) +1. 
\end{equation*}
By iterating the definition, we obtain the upper bound
\begin{equation*}
\max_{1\leq n \leq N} p(n) \leq  20 \max\Big( (\gamma-1)^{-1}, \Big( \frac{2}{\gamma}-1 \Big)^{-1} \Big) \gamma^N \max(p,Nd) \leq p^\ast/2. 
\end{equation*}
Turning to the length scales $(L_k)_{k=0}^{\infty}$, we recall that the initial length scale $L_0$ is part of the statement of Theorem \ref{thm:ind_scale}. Due to our assumption, the initial length scale satisfies
\begin{equation}\label{eq:p_main_1}
L_0 \geq C(\rho,m,\tau,\gamma,\beta,d) \max_{1\leq n \leq N}  \Big( n + p(n) \Big)^{\frac{1}{\gamma \beta} + \frac{2}{1-\gamma}},
\end{equation}
where the constant $C$ is as in Theorem \ref{thm:ind_scale}. The remaining length scales $L_k$, where $k\geq 1$, are then defined through the (forward) recursion $L_{k}= L_{k-1}^\gamma$. 

We now turn to the localization lengths $(m_k)_{k=0}^\infty$. With $m>0$ as in the statement of the theorem, we define 
\begin{equation*}
m_0 = 2m \qquad \text{and} \qquad m_{k} = m_{k-1} \big( 1 - 3 L_{k-1}^{-\frac{1-\tau}{2}} \big) \big( 1 - 250 N^2 L_{k-1}^{1-\tau \gamma}\big) \quad \text{for all } k \geq 1. 
\end{equation*} 
Due to our choice of $m_0$, the length scales $(L_k)_{k=0}^\infty$, and the (large) constant $C=C(\rho,m,\tau,\gamma,\beta,d)$, all three factors in the recursion formula are positive and $m_k \geq m$ for all $k\geq 0$. \\

We now prove by a induction on $k\geq 0$ that 
\begin{equation}\label{eq:p_main_2}
\inf_{\bx_0 \in \bZ^{nd}} \mathbb{P} \{\Lambda_{L}^{(n)}(\bx_0)\textup{ is $m_k$-localizing for $H$}\} \geq 1 - L_k^{-p(n)} \qquad \text{for all } 1 \leq n \leq N.  
\end{equation}

\emph{Base case (in the length scale): $k=0$.}
Due to our assumption on the disorder parameter $\lambda$, we have for all $1\leq n \leq N$ that 
\begin{equation*}
\lambda \geq 2n d \| \rho\|_\infty (1+e^{m_0}) (n!)^2 (2L_0+1)^{2nd} L_0^{p(n)}. 
\end{equation*}
Thus, \eqref{eq:p_main_2} with $k=0$ follows from Proposition \ref{prop:initial_scale}. 

\emph{Induction step (in the length scale): $k-1\rightarrow k$.}   Due to our choice of $(p(n))_{n=1,\hdots,N}$, the condition \eqref{equ:pN_cond} in Theorem \ref{thm:ind_scale} is satisfied. The induction step then follows from the lower bound \eqref{eq:p_main_1} and Theorem \ref{thm:ind_scale} with $(\ell,L,m,M)$ replaced by $(L_{k-1},L_k,m_{k-1},m_k)$. \\

This completes the proof of our claim \eqref{eq:p_main_2}. The claim \eqref{eq:p_main_2} almost yields \eqref{eq:AL} in Theorem \ref{theorem:AL}, except that the length scale $L$ is currently restricted to the discrete sequence $(L_k)_{k=0}^\infty$. This restriction can essentially be removed as in Section \cite[Section 4.3]{EK16}.  Since \cite[Section 4.3]{EK16} uses exponential instead of polynomial tails, we mention that $p$ in the derivation of \eqref{eq:p_main_2} has to be replaced by $p+4Nd$, but omit all other details.
\end{proof}

\begin{appendix}

\section{Non-decreasing rearrangement and properties of covers}
Before we prove the covering properties, we provide a more convenient representation of the symmetrized distance. 
\begin{definition}[Non-decreasing rearrangement]\label{definition:non_decreasing}
Let \( \bx \in \mathbb{Z}^N \). We define \( \widehat{\bx} \in \bZ^N \) as the unique vector satisfying \( \widehat{\bx}_j \leq \widehat{\bx}_{j+1} \) for all \( j=1,\hdots,N-1 \) and \( \widehat{\bx} = \pi \bx \) for some \( \pi \in S_N \). We call \( \widehat{\bx} \) the non-decreasing rearrangement of \( \bx \). Furthermore, we call \( \bx \in \bZ^N  \) non-decreasing if \( \bx = \widehat{\bx} \). 
\end{definition}
Definition \ref{definition:non_decreasing} heavily depends on the natural order on \( \bZ\) and has no (exact) analogue in \( \bZ^d \) for \( d \geq 2 \). The next lemma characterizes the symmetrized distance of \( \bx \) and \( \by \) in terms of their non-decreasing rearrangements. 

\begin{lemma}\label{lemma:rearrangement}
Let \( N \geq 2 \) and let \( \bx,\by \in \bZ^N \). Then, it holds that 
\begin{equation}\label{eq:symmetrized_non_decreasing}
d_S(\bx,\by) = \| \widehat{\bx} - \widehat{\by} \|_{\infty}. 
\end{equation}
\end{lemma}
\begin{proof}
Due to the permutation invariance of the identity \eqref{eq:symmetrized_non_decreasing}, we may assume that \( \bx \) and \( \by \) are non-decreasing. Then, the identity \eqref{eq:symmetrized_non_decreasing} is equivalent to 
\begin{equation}\label{eq:symmetrized_non_decreasing_2}
\| \bx - \pi \by \|_\infty \geq \| \bx - \by \|_\infty \qquad \forall \pi \in S_N. 
\end{equation}
We proceed by induction on the particle number \( N \geq 2 \). \\

\noindent\emph{Base case: \(  N= 2 \).} Since \eqref{eq:symmetrized_non_decreasing_2} clearly holds if \( \pi \) equals the identity, the base case reduces to 
\begin{equation*}
\max( |x_1 - y_2|, |x_2 - y_1|) \geq \max( |x_1 - y_1|, |x_2 - y_2|). 
\end{equation*}
Using the symmetry in \( \bx \) and \( \by \), we can further reduce \eqref{eq:symmetrized_non_decreasing_2} to 
\begin{equation}\label{eq:symmetrized_non_decreasing_p1}
\max( |x_1 - y_2|, |x_2 - y_1|) \geq \max( x_1 - y_1, x_2 - y_2). 
\end{equation}
Using that \( \bx \) and \( \by \) are non-decreasing,  \eqref{eq:symmetrized_non_decreasing_p1} follows directly from 
\begin{equation*}
x_1 - y_1 \leq x_2 - y_1 \qquad \text{and} \qquad x_2 - y_2 \leq x_2 - y_1. 
\end{equation*}
\emph{Induction step: \( N-1\rightarrow N \).} We further split the induction step into two cases. 

\noindent \emph{Case 1: \( \pi(N) = N \).} This case easily follows from the induction hypothesis. Indeed, let \( \widetilde{\bx} = (x_1,\hdots,x_{N-1}) \), \( \widetilde{\by} = ( y_1,\hdots,y_{N-1} ) \), and let \( \widetilde{\pi}\in S_{N-1} \) be the restriction of \( \pi \) to \( \{1,\hdots,N-1\} \). Using the induction hypothesis, we obtain that 
\begin{equation*}
\| \bx - \pi \by \|_\infty = \max( \| \widetilde{\bx} - \widetilde{\pi} \widetilde{\by} \|_\infty, |x_N - y_N|) 
\geq  \max( \| \widetilde{\bx} - \widetilde{\by} \|_\infty, |x_N - y_N|)
= \| \bx - \by \|_\infty. 
\end{equation*}
\emph{Case 2: \( \pi(N) \neq N \).} In this case, the largest entry of \( \pi \by \) is not in the \( N\)-th coordinate. The idea is to iterate the base case (\(N=2\)) to move the largest entry into the \( N\)-th coordinate, which then allows us to use Case 1. To be more precise, let \( j=\pi(N) \). For any \( j < k \leq N \), we let \( \tau_k \in S_N \) be the transposition of \( k-1 \) and \( k \). We set \( \pi_j := \pi \) and \( \pi_k := \tau_k \circ \pi_{k-1} \) for any \( j < k \leq N \). From \( \pi(N) = j \), it follows that \( \pi_k(N) = k \). Since \( \by \) is non-decreasing, we easily see from the definitions that 
\begin{equation*}
(\pi_{k-1} \by)_{k-1} = (\pi_j \by)_j = (\pi \by)_j = y_N \geq (\pi_{k-1} \by)_{k}. 
\end{equation*}
Applying the base case to the \( k-1\)-th and \( k \)-th coordinate, this implies for all \( j < k \leq N \) that 
\begin{equation}\label{eq:symmetrized_non_decreasing_p2}
\| \bx - \pi_{k-1} \by \|_\infty \geq \| \bx - (\tau_k \circ \pi_{k-1}) \by \|_\infty = \| \bx - \pi_{k} \by \|_\infty. 
\end{equation}
By iterating \eqref{eq:symmetrized_non_decreasing_p2}, we obtain that 
\begin{equation}\label{eq:symmetrized_non_decreasing_p3}
\| \bx - \pi \by \|_\infty  \geq \| \bx - \pi_N \by \|_\infty. 
\end{equation}
Since \( \pi_N(N) = N \), \eqref{eq:symmetrized_non_decreasing_p3} and Case 1 imply the desired estimate \eqref{eq:symmetrized_non_decreasing_2}. 
\end{proof}

After these preparations, we can now prove the desired properties of covers. 
\begin{proof}[Proof of Lemma \ref{lemma:cover}:]
We first prove \eqref{item:cover_1}. Without accounting for the permutation invariance of the centers, we only obtain \( \# C_{L,\ell} \leq \# \Lambda_L^{(N)}(\bx_0) \leq N! (2L+1)^{Nd} \). However, since \( \Lambda_l^{(N)}(\ba_1) = \Lambda_l^{(N)}(\ba_2) \) if (and only if) \(\ba_1 =\pi \ba_2 \) for some \( \pi \in S_N \), it holds that 
\begin{equation}
\# C_{L,\ell} \leq \Big \{ a \in \bZ^d \colon \| \ba - \bx_0 \|_{\infty} \leq L \Big \} \leq (2L+1)^{Nd}. 
\end{equation}
It remains to prove \eqref{item:cover_2}. We let \( \bx \in \Lambda_L(\bb) \). By the permutation invariance of \eqref{eq:cover}, both in terms of the center and the set itself, we can assume that \( \bx \) and \( \bb \) are non-decreasing. If \( d_S(\bb,\bx) =\| \bb - \bx \|_\infty \leq L-\ell \), we can simply choose \( \ba:= \bx \in \Xi_{L,\ell} \). Unfortunately, if \(  L-\ell < \| \bb - \bx \|_\infty \leq L \), then \( \bx \not \in \Xi_{L,\ell} \) and this choice is not admissible. Instead, we then choose \( \ba \) as the truncation of \( \bx \) with respect to \( \bb \), i.e., we define
\begin{equation*}
a_j := b_j + \operatorname{sign}(x_j-b_j) \cdot \min( |x_j-b_j|, L-\ell).
\end{equation*}
An equivalent expression is given by
\begin{equation*}
a_j = \max( b_j - (L-\ell), \min( x_j, b_j + (L-\ell))). 
\end{equation*}
Since both \( \bx \) and \( \bb \) are non-decreasing, \( \ba \) is also non-decreasing. We now want to show that \( \bx \in \Lambda_\ell^{\Lambda_L(\bb),\ell}(\ba) \), which is equivalent to proving that 
\begin{equation}\label{eq:cover_p1}
\by \in \Lambda_L(\bb) \backslash \Lambda_\ell(\ba) \quad \Longrightarrow \quad d_S(\by,\bx) > \ell. 
\end{equation}
Using Lemma \ref{lemma:rearrangement}, we have for any \( \by \in \Lambda_L(\bb) \backslash \Lambda_\ell(\ba) \) that 
\begin{equation}\label{eq:cover_p2}
d_S(\bx,\by) = \| \bx - \widehat{\by} \|_\infty, \quad d_S(\bb,\widehat{\by})= \| \bb - \widehat{\by} \|_\infty \leq L, \quad \text{and} \quad d_S(\ba,\by)= \| \ba - \widehat{\by} \|_\infty > \ell. 
\end{equation}
We may therefore replace \( \by \) in \eqref{eq:cover_p1} by \( \widehat{\by} \), and thus assume that \( \by \) is non-decreasing. After this reduction, it remains to prove for all (non-decreasing) \( \by \in \bZ^N \) that 
\begin{equation}\label{eq:cover_p3}
\| \by - \bb \|_\infty \leq L, ~ \| \by - \ba \|_\infty > \ell \quad \Longrightarrow \quad \| \by - \bx \|_\infty > \ell. 
\end{equation}
To prove \eqref{eq:cover_p3}, let \( 1 \leq j \leq N \) be such that \( |y_j - a_j|> \ell \). If \( |x_j - b_j|\leq L-\ell \), then \( a_j = x_j \) and hence \( |y_j - x_j|> \ell \). If \( |x_j - b_j | > L-\ell \), we assume without loss of generality that \( x_j > b_j + L-\ell \). Then, \( a_j = b_j + L-\ell \). Since \( |y_j -a_j | > \ell \) and \( y_j \leq b_j + L \), it follows that \( y_j < a_j - \ell \). This implies
\begin{equation*}
x_j - y_j > x_j - (a_j - \ell) \geq \ell. 
\end{equation*}
A visualization of the last step is provided in Figure \ref{figure:proof_step}. 
\end{proof}
\begin{remark}
The most severe difficulty in extending the covering lemma to spatial dimensions \( d \geq 2 \) lies in the analog of \eqref{eq:cover_p2}. In one spatial dimension, the symmetrized distance of \( \by \) and \( \ba,\bb \), or \( \bx \) is realized by the \emph{same} permutation, namely  the non-decreasing rearrangement. Unfortunately, this does not seem to hold in dimensions \( d \geq 2 \) and the different symmetrized distance may be realized by different permutations. This dramatically increases the geometric complexity behind the covering lemma. 
\end{remark}

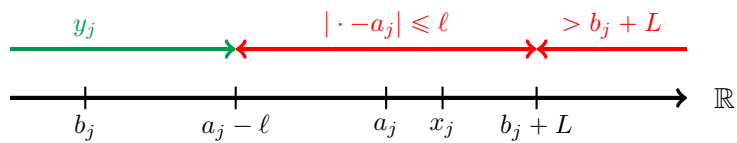
\begin{figure}[h]
\tikzmath{\smallL= 2; \bigL =6; \oset=0.15; \h = 0.65; }
\begin{tikzpicture}

\draw[ultra thick,->] (-1,0)--(\bigL+2,0);

\draw[thick] (0,-\oset) -- (0,\oset);
\draw[thick] (\bigL,-\oset) -- (\bigL,\oset);
\draw[thick] (\bigL-\smallL,-\oset) -- (\bigL-\smallL,\oset);
\draw[thick] (\bigL-\smallL+0.75,-\oset) -- (\bigL-\smallL+0.75,\oset);
\draw[thick] (\bigL-2*\smallL,-\oset) -- (\bigL-2*\smallL,\oset);

\node at (\bigL+2.5,0) {$\mathbb{R}$};
\node at (0,-.4) {\small{$b_j$}}; 
\node at (\bigL,-.4) {\small{$b_j+L$}}; 
\node at (\bigL-\smallL,-.4) {\small{$a_j$}}; 
\node at (\bigL-2*\smallL,-.4) {\small{$a_j-\ell$}}; 
\node at  (\bigL-\smallL+0.75,-0.4) {\small{$x_j$}};

\draw[ultra thick,<->,red] (\bigL- 2 *\smallL,\h) -- (\bigL,\h);
\draw[ultra thick,<-,red] (\bigL,\h) -- (\bigL+2,\h);
\draw[ultra thick,->,color=ForestGreen] (-1,\h) -- (\bigL-2*\smallL,\h);

\node[above] at (0,\h) {\textcolor{ForestGreen}{\small{$y_j$}}};
\node[above] at (\bigL-\smallL,\h) {\textcolor{red}{\small{$|\cdot - a_j| \leq \ell $}}};
\node[above] at (\bigL+1,\h) {\textcolor{red}{\small{$ > b_j +L  $}}};
\end{tikzpicture}
\caption{\small{This figure illustrates the last step in the proof of Lemma \ref{lemma:cover}. The two conditions $|y_j-a_j| > \ell$ and $|y_j-b_j|\leq L$ } imply that $y_j$ cannot lie inside the red areas. The remaining possible values of $y_j$, which are displayed in green, are at a distance greater than $\ell$ from $x_j$. }
\label{figure:proof_step}
\end{figure}

\end{appendix}

\bibliography{BM_library_intro}
\bibliographystyle{alpha}

\end{document}